\documentclass[english,runningheads,11pt]{llncs}
\usepackage{graphicx,amssymb,amsmath,amssymb}
\usepackage[linesnumbered, vlined, ruled]{algorithm2e}
\usepackage[absolute]{textpos}

\usepackage[in]{fullpage}

\def\calP{\mathcal{P}}
\def\calT{\mathcal{T}}
\def\calV{\mathcal{V}}

\def\calM{\mathcal{M}}
\def\st{$s$-$t$}

\newtheorem{observation}{Observation}

\begin{document}

\title{A Divide-and-Conquer Algorithm for Two-Point $L_1$ Shortest Path Queries in Polygonal Domains
}
\author{
Haitao Wang
}
\institute{
Department of Computer Science\\
Utah State University, Logan, UT 84322, USA\\
\email{haitao.wang@usu.edu}
}

\maketitle

\pagenumbering{arabic}
\setcounter{page}{1}

\vspace*{-0.1in}
\begin{abstract}
Let $\calP$ be a polygonal domain of $h$ holes and $n$ vertices. We study the problem of constructing a data structure that can compute a shortest path between $s$ and $t$ in $\calP$ under the $L_1$ metric for any two query points $s$ and $t$. To do so, a standard approach is to first find a set of $n_s$ ``gateways'' for $s$ and a set of $n_t$ ``gateways'' for $t$ such that there exist a shortest \st\ path containing a gateway of $s$ and a gateway of $t$, and then compute a shortest \st\ path using these gateways. Previous algorithms all take quadratic $O(n_s\cdot n_t)$ time to solve this problem. In this paper, we propose a divide-and-conquer technique that solves the problem in $O(n_s + n_t \log n_s)$ time. As a consequence,
we construct a data structure of $O(n+(h^2\log^3 h/\log\log h))$ size in $O(n+(h^2\log^4 h/\log\log h))$ time such that each query can be answered in $O(\log n)$ time.
\end{abstract}


\section{Introduction}
\label{sec:intro}

Let $\calP$ be a polygonal domain of $h$ holes with a total of $n$
vertices, i.e., there is an outer simple polygon containing $h$
disjoint holes and each hole itself is a simple polygon. If $h=0$, then $\calP$ becomes a simple polygon. For any two points $s$ and $t$, an {\em $L_1$ shortest path} from $s$ to $t$ in $\calP$ is a path connecting $s$ and $t$ with the minimum length under the $L_1$ metric. Note that the edges of the path can have arbitrary slopes but their lengths are measured by the $L_1$ metric.

We consider the two-point $L_1$ shortest path query problem: Construct
a data structure for $\calP$ that can compute an $L_1$ shortest path
in $\calP$ for any two query points $s$ and $t$. To do so, a standard approach is to first find a set of $n_s$ ``gateways'' for $s$ and a set of $n_t$ ``gateways'' for $t$ such that there exist a shortest \st\ path containing a gateway of $s$ and a gateway of $n_t$, and then compute a shortest \st\ path using these gateways. Previous algorithms~\cite{ref:ChenTw16,ref:ChenSh00}  all take quadratic $O(n_s\cdot n_t)$ time to solve this problem. In this paper, we propose a divide-and-conquer technique that solves the problem in $O(n_s + n_t \log n_s)$ time.

As a consequence, we construct a data structure of $O(n+(h^2\log^3 h/\log\log h))$ size in $O(n+(h^2\log^4 h/\log\log h))$ time such that each query can be answered in $O(\log n)$ time\footnote{Throughout the paper, unless otherwise stated, when we
say that the query time of a data structure is $O(T)$, we mean that
the shortest path length can be computed in $O(T)$ time and an actual
shortest path can be output in additional linear time in the number of edges of
the path.}.
Previously, Chen et al.~\cite{ref:ChenSh00} built a data structure of $O(n^2\log n)$ size
in $O(n^2\log^2 n)$ time that can answer each query in $O(\log^2 n)$
time. Later Chen et al.~\cite{ref:ChenTw16} achieved $O(\log n)$
time queries by building a data structure of $O(n+h^2\cdot \log h\cdot
4^{\sqrt{\log h}})$ space in $O(n+h^2\cdot \log^2 h\cdot 4^{\sqrt{\log
h}})$ time.
The preprocessing complexities of our result improve the
previous work~\cite{ref:ChenTw16} by a super polylogarithmic factor.
More importantly, our divide-and-conquer technique may be interesting in its own right.

\subsection{Related Work}

Better results exist for certain special cases of the problem. If $\calP$ is a simple polygon, then a shortest path in $\calP$ with minimum Euclidean length is also an $L_1$ shortest path~\cite{ref:HershbergerCo94}, and thus by using the data structure in~\cite{ref:GuibasOp89,ref:HershbergerA91} for the Euclidean metric, one can build a data structure in $O(n)$ time and space that can answer each query in $O(\log n)$ time; recently Bae and Wang~\cite{ref:BaeL118} proposed a simpler approach that can achieve the same performance.
If $\calP$ and all holes of it are rectangles whose edges are all axis-parallel, then ElGindy and Mitra~\cite{ref:ElGindyOr94} constructed a data structure of $O(n^2)$ size in $O(n^2)$ time that supports $O(\log n)$ time queries.

Better results are also known for {\em one-point queries} in the $L_1$ metric~\cite{ref:ChenCo19,ref:ClarksonRe87,ref:ClarksonRe88,ref:InkuluPl09,ref:MitchellAn89,ref:MitchellL192}, i.e., $s$ is fixed in the input and only $t$ is a query point. In particular, Mitchell~\cite{ref:MitchellAn89,ref:MitchellL192} built a data structure of $O(n)$ size in $O(n\log n)$ time that can answer each such query in $O(\log n)$ time. Later Chen and Wang~\cite{ref:ChenCo19} reduced the preprocessing time to $O(n+h\log h)$ if $\calP$ is already triangulated (which can be done in $O(n\log n)$ or $O(n+h\log^{1+\epsilon}h)$ time for any $\epsilon>0$~\cite{ref:Bar-YehudaTr94,ref:ChazelleTr91}), while the query time is still $O(\log n)$.

The Euclidean counterparts have also been studied. For one-point
queries, Hershberger and Suri~\cite{ref:HershbergerAn99} built a
shortest path map of $O(n)$ size with $O(\log n)$ query time and the
map can be built in $O(n\log n)$ time and space. For two-point
queries, Chiang and Mitchell~\cite{ref:ChiangTw99} built a data
structure of $O(n^{11})$ size that can support $O(\log n)$ time
queries, and they also built a data structure of $O(n+h^5)$ size with
$O(h\log n)$ query time. Other results with tradeoff between
preprocessing and query time were also proposed in~\cite{ref:ChiangTw99}.
Also, Chen et al.~\cite{ref:ChenOn01} showed that with $O(n^2)$ space one can answer each two-point query in $O(\min\{|Q_s|,|Q_t|\}\cdot \log n)$ time, where $Q_s$ (resp., $Q_t$) is the set of vertices of $\calP$ visible to $s$ (resp., $t$). Guo et al.~\cite{ref:GuoSh08} gave a data structure of $O(n^2)$ size that can support $O(h\log n)$ time two-point queries.

\subsection{Our Techniques}
\label{sec:approach}

We follow a similar scheme as in~\cite{ref:ChenTw16,ref:ChenSh00}, using a ``path-preserving'' graph $G$ proposed by Clarkson et al.~\cite{ref:ClarksonRe87,ref:ClarksonRe88} to determine a set $V_g(q)$ of $O(\log n)$ points (called ``gateways'') for each query point $q\in \{s,t\}$, such that there exists an $L_1$ shortest \st\ path that contains a gateway in $V_g(s)$ and a gateway in $V_g(t)$. To find a shortest \st\ path, the main difficulty is to solve the following sub-problem. Let $\pi(p,q)$ denote a shortest path between two points $p$ and $q$ in $\calP$, and let $d(p,q)$ denote the length of the path.
Suppose that the gateways of $s$ (resp., $t$) are formed as a cycle around $s$ (resp., $t$), e.g., see Fig.~\ref{fig:idea}, such that there is a shortest \st\ path containing a gateway of $s$ and a gateway of $t$. The point $s$ is visible to each gateway $p$ in $V_g(s)$, and thus $d(s,p)$ can be obtained in $O(1)$ time for any $p\in V_g(s)$. The same applies to $t$. Also suppose in the preprocessing we have computed $d(p,q)$ for any  $p\in V_g(s)$ and any $q\in V_g(t)$. The goal of the problem is to find $p\in V_g(s)$ and $q\in V_g(t)$ such that the value $d(s,p)+d(p,q)+d(q,t)$ is minimized, so that a shortest \st\ path contains both $p$ and $q$.

To solve the sub-problem, a straightforward method is to try all pairs
of $p$ and $q$ with $p\in V_g(s)$ and  $q\in V_g(t)$, which is the
approach used in both algorithms in~\cite{ref:ChenTw16,ref:ChenSh00}.
This takes $O(n_s\cdot n_t)$ time, where $n_s=|V_g(s)|$ and $n_t=|V_g(t)|$. In~\cite{ref:ChenSh00}, both $n_s$ and $n_t$ are bounded by $O(\log n)$, which results in an $O(\log^2 n)$ time query algorithm. In~\cite{ref:ChenTw16}, both $n_s$ and $n_t$ are reduced to $O(\sqrt{\log n})$, and thus the query time becomes $O(\log n)$, by using a larger ``enhanced graph'' $G_E$ (than the original graph $G$). More specifically, the size of $G$ is $O(n\log n)$ while the size of $G_E$ is $O(n\sqrt{\log n}2^{\sqrt{\log n}})$ (which is further reduced to $O(h\sqrt{\log h}2^{\sqrt{\log h}})$ by other techniques~\cite{ref:ChenTw16}).

\begin{figure}[t]
\begin{minipage}[t]{0.49\linewidth}
\begin{center}
\includegraphics[totalheight=1.2in]{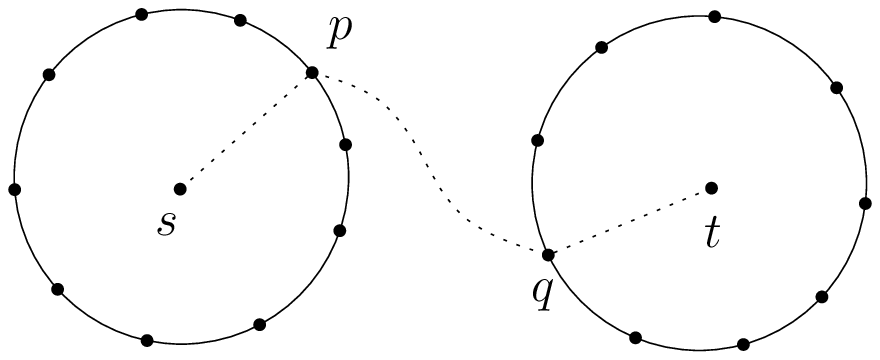}
\caption{\footnotesize
Illustrating the gateways of $s$ and $t$ and a shortest \st\ path.}
\label{fig:idea}
\end{center}
\end{minipage}
\hspace{0.05in}
\begin{minipage}[t]{0.49\linewidth}
\begin{center}
\includegraphics[totalheight=1.4in]{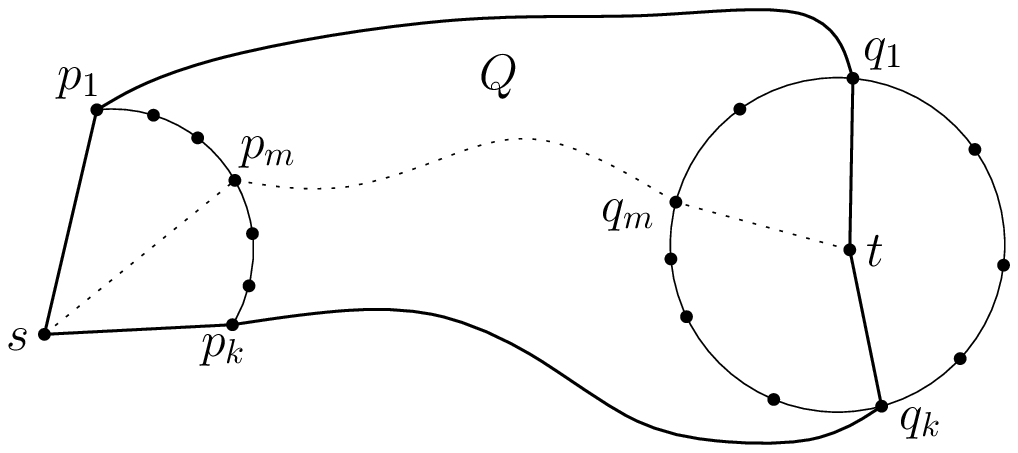}
\caption{\footnotesize
Illustrating our divide-and-conquer scheme.}
\label{fig:idea1}
\end{center}
\end{minipage}
\vspace*{-0.15in}
\end{figure}

Our main contribution is to develop an $O(n_s+n_t\log n_s)$ time
algorithm for solving the above sub-problem. To this end, we explore
the geometric structures of the problem and propose a
divide-and-conquer technique, which can be roughly described as
follows. For simplicity, suppose we only consider one piece of the
gateway cycle of $s$ (e.g., those in the first quadrant of $s$)
and order the gateways of $s$ on that piece by $p_1, p_2, \ldots, p_k$
(e.g., see Fig.~\ref{fig:idea1}). Then, in a straightforward way,
for $p_1$, we find a gateway, denoted by $q_1$, of $t$ that minimizes the value
$d(p_1,q)+d(q,t)$ for all $q\in V_g(t)$. Similarly, we find such a
gateway $q_k$ of $t$ for $p_k$. Let $P_1$ be the \st\ path
$\overline{sp_1}\cup \pi(p_1,q_1)\cup \overline{q_1t}$. Similarly, let $P_2$ be
the path $\overline{sp_k}\cup \pi(p_k,q_k)\cup \overline{q_kt}$. In the
``ideal'' situation, the two paths do not intersect except at $s$ and $t$, and they together
form a cycle enclosing a plane region $Q$ that contains all gateways
$p_1,p_2,\ldots,p_k$ (e.g., see Fig.~\ref{fig:idea1}), and let
$V'_g(t)$ be the gateways of $t$ that are also contained in $Q$. The
next step is to process the median gateway $p_{m}$ of $s$ with
$m=\frac{k}{2}$. The key observation is that we only need to consider
the gateways in $V'_g(t)$ instead of all the gateways of $t$, i.e., if a
shortest \st\ path contains $p_{m}$, then there must be a shortest
\st\ path containing $p_{m}$ and a gateway in $V'_g(t)$. In this way,
we only need to find the point, denoted by $q_{m}$, that minimizes the value
$d(p_{m},q)+d(q,t)$ for all $q\in V'_g(t)$. Further, in the ``ideal''
situation, the path $P_m=\overline{sp_m}\cup \pi(p_m,q_m)\cup \overline{q_mt}$
is inside the region $Q$ and divides $Q$ into two sub-regions  (e.g., see Fig.~\ref{fig:idea1}). We then proceed on the two sub-regions recursively.

\begin{figure}[t]
\begin{minipage}[t]{0.49\linewidth}
\begin{center}
\includegraphics[totalheight=1.2in]{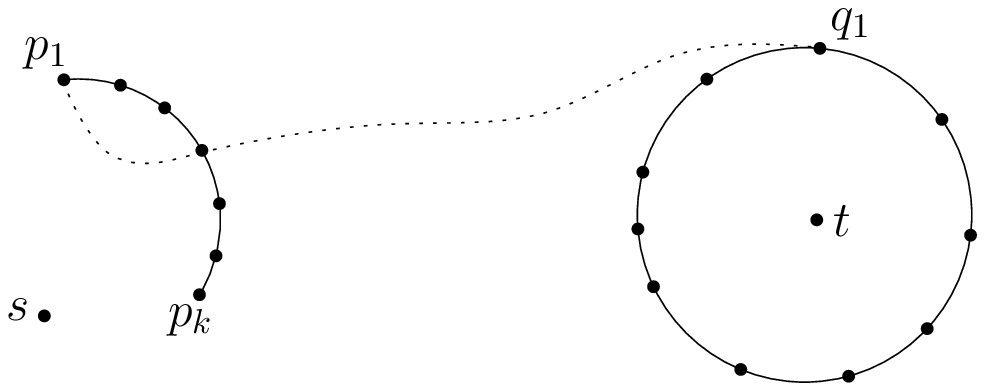}
\caption{\footnotesize
Illustrating a non-ideal situation: The shortest path from $p_1$ to $q_1$ crosses the gateway cycle of $s$.}
\label{fig:dif1}
\end{center}
\end{minipage}
\hspace{0.05in}
\begin{minipage}[t]{0.49\linewidth}
\begin{center}
\includegraphics[totalheight=1.4in]{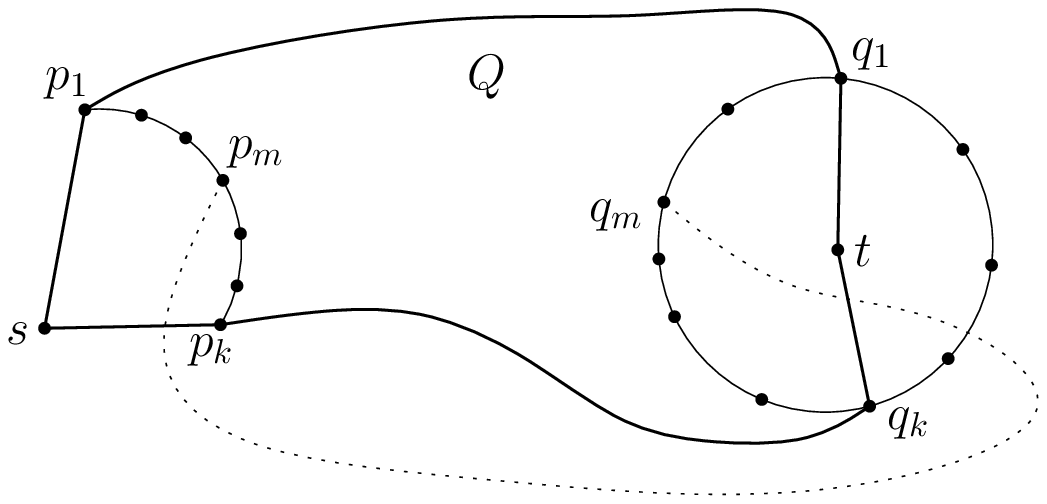}
\caption{\footnotesize
Illustrating a non-ideal situation: The shortest path from $p_m$ to $q_m$ is not inside the region $Q$.}
\label{fig:dif2}
\end{center}
\end{minipage}
\vspace*{-0.15in}
\end{figure}

The above exhibits our algorithm in an ``ideal'' situation. Our major
effort is to deal with the ``non-ideal'' situations. For examples,
what if the path $P_1$ divides the cycle piece of $s$ into two parts (e.g., see Fig.~\ref{fig:dif1}), what if the path $P_m$ is not in the region $Q$ (e.g., see Fig.~\ref{fig:dif2}),  what if $q_1=q_k$, etc. 

Note that our divide-and-conquer scheme may be somewhat similar to
that for two-vertex shortest path queries in planar graphs,
e.g.,~\cite{ref:ChenSh00Xu,ref:DjidjevEf96}. However, a main difference
is that in the planar graph case the query vertices are both from the
input graph and the gateways are already known for each vertex (more specifically, the gateways in the planar graph case are the ``border vertices'' of the subgraphs in the decomposition of the input graph by separators), and thus one can compute certain information for the gateways in the preprocessing (many other techniques for shortest path queries in planar graphs, e.g.,~\cite{ref:FakcharoenpholPl06,ref:GawrychowskiBe18,ref:MozesEx12}, also rely on this), while in our problem the gateways are only determined
``online'' during queries because both query points can be
anywhere in $\calP$.
This causes us to develop different techniques to tackle the problem
(especially to resolve the non-ideal situations).
To the best of our knowledge, this is the first time such a divide-and-conquer method is applied to geometric setting for shortest path queries using gateways


With the above $O(n_s+n_t\log n_s)$ time algorithm, if both $n_s$ and $n_t$ are bounded by $O(\log n)$, we can only obtain an $O(\log n \log\log n)$ time query algorithm. To reduce the time to $O(\log n)$, we borrow some idea from the previous work~\cite{ref:ChenTw16} to construct a larger graph $G_1$, so that we can guarantee $n_s=O(\log n)$ and $n_t=O(\log n/\log\log n)$, which leads to an $O(\log n)$ time query algorithm. The size of $G_1$ is only $O(n\log^2 n/\log\log n)$, which is slightly larger than the original $O(n\log n)$-sized graph $G$~\cite{ref:ChenSh00,ref:ClarksonRe87,ref:ClarksonRe88} and much smaller than the $O(n\sqrt{\log n}2^{\sqrt{\log n}})$-sized enhanced graph $G_E$ in~\cite{ref:ChenTw16}. Further, by the techniques similar to those used in~\cite{ref:ChenTw16}, we can reduce the graph size to $O(h\log^2 h/\log\log h)$.

We stress that although the overall preprocessing of our data structure only improves the previous work~\cite{ref:ChenTw16} by roughly a factor of $4^{\sqrt{\log h}}$ (super-polylogarithmic but sub-polynomial), our contribution is more on the $O(n_s+n_t\log n_s)$ time divide-and-conquer query algorithm, which is nearly a linear factor improvement over the previous $O(n_s\cdot n_t)$ time algorithms~\cite{ref:ChenTw16,ref:ChenSh00} and is the first-known sub-quadratic time algorithm in the number of gateways of the query points.


The rest of the paper is organized as follows. In Section~\ref{sec:pre}, we define notation and review some previous work. In Section~\ref{sec:subproblem}, we solve the sub-problem discussed above. In Section~\ref{sec:overall}, we present our overall result.
For ease of exposition, we make a general position assumption that no two vertices of $\calP$ including $s$ and $t$ have the same $x$- or $y$-coordinate. Unless otherwise stated, ``length'' always refers to $L_1$ length and ``shortest paths'' always refers to $L_1$ shortest paths.

\section{Preliminaries}
\label{sec:pre}

We introduce some notation and concepts, some of which are borrowed from the previous work~\cite{ref:ChenTw16,ref:ChenSh00,ref:ClarksonRe87,ref:ClarksonRe88}.

Two points $p$ and $q$ are {\em visible} to each other if the line
segment $\overline{pq}$ is in $\calP$. For a point $p$ and a vertical
line segment $l$ in $\calP$, if there is a point $q\in l$ such that
$\overline{pq}$ is horizontal and is in $\calP$, then we say that $p$
is {\em horizontally visible} to $l$ and we call $q$ the {\em horizontal
projection} of $p$ on $l$.

For any point $p$ in the plane, we use $x(p)$ and $y(p)$ to denote its
$x$- and $y$-coordinates, respectively. In the paper, when we talk
about a relative position (e.g., left, right, above, below,
northeast) of two geometric objects (e.g., lines, points),
unless there is a ``strictly'', it always includes the tie case. For
example, if we say that a point $p$ is to the northeast of another
point $q$, then we mean $x(p)\geq x(q)$ and $y(p)\geq y(q)$.
Similarly, if we say that a point $p$ is to the left of a vertical
line $l$, then either $p$ is strictly to the left of $l$ or $p$ is on
$l$.

For a path $\pi$ in $\calP$, we use $|\pi|$ to denote its length.
For two points $p$ and $q$ in $\calP$, we use $\pi(p,q)$ to denote
a shortest path from $p$ to $q$ and define $d(p,q)=|\pi(p,q)|$.
For a segment $\overline{pq}$, we use $|\overline{pq}|$ to denote the
length of $\overline{pq}$.
A path in $\calP$ is {\em $x$-monotone} if its intersection with any
vertical line is either empty or connected. The {\em $y$-monotone} is
defined similarly. If a path is both $x$-monotone and $y$-monotone,
then it is {\em $xy$-monotone}. Note that an $xy$-monotone path in
$\calP$ is a shortest path. Also, if there is an $xy$-monotone path
between $p$ and $q$ in $\calP$, then
$d(p,q)=|\overline{pq}|$ (although $p$ may not be visible to $q$).

Let $\calV$ denote the set of all vertices of $\calP$. To
differentiate from the vertices and edges in some graphs we define later, we
often refer to the vertices of $\calP$ as {\em polygon vertices} and
the edges of $\calP$ as {\em polygon edges}.
Let $\partial \calP$ denote the boundary of $\calP$ (including the
boundaries of all the holes). For any point $p\in \calP$, if we shoot a
ray rightwards from $p$, let $p^r$ denote the first point of
$\partial\calP$ hit by the ray and call it the {\em rightward
projection} of $p$ on $\partial\calP$. Similarly, we can define the
leftward, upward, downward projections of $p$ and denote them by
$p^l$, $p^u$, $p^d$, respectively.

\paragraph{A ``path-preserving'' graph $G$.}
Clarkson et al.~\cite{ref:ClarksonRe87} proposed a
graph $G$ for computing $L_1$ shortest paths in $\calP$.
We sketch the graph $G$ below, since our algorithm will use a modified version of it.

To define $G$, there are two types of {\em Steiner points}.
For each vertex of $\calP$, its four projections
on $\partial\calP$ are {\em type-1} Steiner points. Hence, there are
$O(n)$ Steiner points on $\partial\calP$. The {\em type-2 Steiner}
points are defined on {\em cut-lines}, which can be organized into a binary tree $\calT$,
called {\em the cut-line tree}. Each node $u$ of $\calT$ corresponds to a set
$\calV(u)$ of vertices of $\calP$ and stores a cut-line $l(u)$ that is a vertical
line through the median $x$-coordinate of all vertices of $\calV(u)$.
If $u$ is the root, then $\calV(u)=\calV$. In general, for the left (resp., right) child $v$
of $u$, $\calV(v)$ consists of all vertices of $\calV(u)$ to the left (resp.,
right) of $l(u)$. For each node $u\in \calT$ and each
vertex $p$ of $\calV(u)$, if $p$ is horizontally visible to $l(u)$, then
the horizontal projection of $p$ on $l(u)$ is a type-2 Steiner point.
Therefore, $l(u)$ has at most $|\calV(u)|$ Steiner points.
Since the total size $|\calV(u)|$ for all $u$ in the same level of $\calT$ is $O(n)$ and the height of
$\calT$ is $O(\log n)$, the total number of type-2 Steiner
points is $O(n\log n)$.

We point out a subtle issue here. If $|\calV(u)|=1$, then $l(u)$ is through the only vertex of $\calV(u)$. Otherwise, if $|\calV(u)|$ is odd, then we slightly change $l(u)$ so that it does not contain a vertex of $\calV(u)$ but still partitions $\calV(u)$ roughly evenly.
In this way, for each polygon vertex $p$, there is a cut-line at the leaf of $\calT$ that contains $p$ and thus $p$ itself is a type-2 Steiner point on the cut-line. Hence, all polygon vertices of $\calV$ are also type-2 Steiner points.

The graph $G$ is thus defined as follows. First of all, the vertex set of $G$ consists of all Steiner points (again polygon vertices are also Steiner points). Hence, it has $O(n\log n)$ nodes.
For the edges of $G$, for each vertex $p$ of $\calP$, if $q$ is a
Steiner point defined by $p$, then $G$ has an edge $\overline{pq}$.
For each polygon edge $e$ of $\calP$, $e$ may contain multiple
Steiner points, and $G$ has an edge connecting each adjacent pair of
them. Further, for each cut-line $l$ and for any two adjacent
Steiner points on $l$, if they are visible to each other, then $G$ has an edge connecting them.

Clearly, $G$ has $O(n\log n)$ nodes and edges.
It was shown in \cite{ref:ClarksonRe87,ref:ClarksonRe88} that for any two polygon vertices of $\calP$, the shortest path between them in the graph $G$ is also a shortest path in $\calP$ (and thus the graph ``preserves'' shortest paths of the polygon vertices of $\calP$).

\paragraph{Gateways.}
In order to answer two-point shortest path queries, Chen et al.~\cite{ref:ChenSh00} ``insert'' the two query points $s$ and $t$ into $G$ by connecting them to some ``gateways''. Intuitively, the gateways would be the vertices of $G$ that connect to $s$ and $t$ respectively if $s$ and $t$ were vertices of $\calP$, and thus they control shortest paths from $s$ to $t$. Specifically, let $V_g(s,G)$ denote the set of gateways for $s$, which has two subsets $V^1_g(s,G)$ and $V^2_g(s,G)$ of sizes $O(1)$ and $O(\log n)$, respectively. We first define
$V^1_g(s,G)$. For each projection point $q$ of $s$ on $\partial\calP$,
if $v_1$ and $v_2$ are the two Steiner points adjacent to $q$ on the edge of $\calP$ containing $q$, then $v_1$ and $v_2$ are in $V^1_g(s,G)$. Since $s$ has four projections on $\partial\calP$, $V^1_g(s,G)$ has at most eight points.
For the set $V^2_g(s,G)$, it is defined recursively on the cut-line
tree $\calT$. Let $u$ be the root of $\calT$. If $s$ is horizontally
visible to the cut-line $l(u)$, then $l(u)$ is called a {\em
projection cut-line} of $s$ and the Steiner point on $l(u)$
immediately above (resp.,  below) the horizontal projection $s'$ of $s$ on
$l(u)$ is a gateway in $V^2_g(s,G)$ if it is visible to $s'$.
Regardless of whether $s$ is horizontally visible to $l(u)$ or not, if
$s$ is to the left (resp., right) of $l(u)$, then we proceed to the
left (resp., right) child of $u$ until we reach a leaf of $\calT$.
Clearly, $s$ has $O(\log n)$ projection cut-lines, which are on a path
from the root to a leaf in $\calT$. Hence, $V^2_g(s,G)$ contains $O(\log n)$ gateways. In a similar way we can define the gateway set $V_g(t,G)$ for $t$.
As will be shown later, for each gateway $p$ of $s$, $\overline{sp}$ is in $\calP$, and thus $d(s,p)=|\overline{sp}|$. The same applies to $t$.

 It is known~\cite{ref:ChenSh00} that if there exists a shortest \st\ path that contains a vertex of $\calP$, then there must exist a shortest \st\ path that contains a gateway of $s$ and a gateway of $t$. On the other hand, if there does not exist any shortest \st\ path containing a vertex of $\calP$, then there must exist a shortest \st\ path $\pi(s,t)$ that is $xy$-monotone and has the following property: either $\pi(s,t)$ consists of a horizontal segment and a vertical segment, or $\pi(s,t)$ consists of three segments: $\overline{ss'}$, $\overline{s't'}$, and $\overline{t't}$, where $s'$ is a vertical (resp., horizontal) projection of $s$ and $t'$ is the horizontal (resp., vertical) projection of $t$ on the same polygon edge. We call such a shortest path as above $\pi(s,t)$ a {\em trivial shortest path}.

\paragraph{A straightforward query algorithm.}
Given $s$ and $t$, we can compute $d(s,t)$ as follows. First, we check whether there exists a trivial shortest \st\ path. As shown in~\cite{ref:ChenSh00}, this can be done in $O(\log n)$ time by using vertical and horizontal ray-shootings, after $O(n\log n)$ time (or $O(n+h\log^{1+\epsilon}h)$ time for any $\epsilon>0$~\cite{ref:Bar-YehudaTr94}) preprocessing to build the vertical and horizontal decompositions of $\calP$.
If yes, then we are done. Otherwise, we compute the gateway sets $V_g(s,G)$ and $V_g(t,G)$ in $O(\log n)$ time after certain preprocessing~\cite{ref:ChenTw16,ref:ChenSh00}. Suppose we have computed $d(u,v)$ for any two vertices $u$ and $v$ of $G$ in the preprocessing, i.e., given $u$ and $v$, $d(u,v)$ can be obtained in constant time. Then, $d(s,t)=\min_{p\in V_g(s,G), q\in V_g(t,G)}(|\overline{sp}|+d(p,q)+|\overline{qt}|)$, which can be computed in $O(\log^2 n)$ time since both $|V_g(s,G)|$ and $|V_g(t,G)|$ are bounded by $O(\log n)$.

\paragraph{The main sub-problem.}
To reduce the query time, since $|V_g^1(s,G)|=O(1)$ and $|V_g^1(t,G)|=O(1)$, the main sub-problem is to determine the value $\min_{p\in V^2_g(s,G), q\in V^2_g(t,G)}(|\overline{sp}|+d(p,q)+|\overline{qt}|)$.
This is the sub-problem we discussed in Section~\ref{sec:approach}. Note that the case $p\in V_g^1(s,G)$ and $q\in V^2_g(t,G)$, or the case $p\in V_g^2(s,G)$ and $q\in V^1_g(t,G)$ can be easily handled in $O(\log n)$ time since both $|V_g^1(s,G)|$ and $|V_g^1(t,G)|$ are $O(1)$.

\section{Solving the Main Sub-Problem}
\label{sec:subproblem}

In this section, we present an $O(n_s+n_t\log n_s)$ time algorithm for our main sub-problem, where $n_s=|V^2_g(s,G)|$ and $n_t=|V^2_g(t,G)|$.

\subsection{Preliminaries}
\label{sec:prep}
We consider the vertices of $G$ as the corresponding points in
$\calP$. Note that although $G$ preserves shortest paths between all
polygon vertices of $\calP$, it may not preserve shortest paths for all vertices
of $G$, i.e., for two vertices $p$ and $q$ of $G$, the shortest path
from $p$ to $q$ in $G$ may not be a shortest path in $\calP$. For this
reason,
as preprocessing, for each vertex $q$ of $G$, we compute a shortest path tree $T(q)$ in $\calP$ from $q$ to all vertices of $G$ using the algorithm in~\cite{ref:MitchellAn89,ref:MitchellL192}, which can be done in $O(n\log^2 n)$ time since $G$ has $O(n\log n)$ vertices.
For each vertex $p$ of $G$, we use $\pi_q(p)$ to denote the path in $T(q)$ from the root $q$ to $p$, which is a shortest path in $\calP$, and we refer to the edge incident to $p$ as the {\em last edge} of $\pi_q(p)$; we explicitly store $d(p,q)$ and the last edge of $\pi_q(p)$.
Note that shortest paths between two points in the $L_1$ metric are in general not unique. However, the shortest path
$\pi_q(p)$ computed by the algorithm in~\cite{ref:MitchellAn89,ref:MitchellL192} has the following property: all vertices of the path other than $p$ and $q$ are polygon vertices of $\calP$.
Doing the above for all vertices $q$ of $G$ takes $O(n^2\log^3 n)$ time and $O(n^2\log^2 n)$ space.

After the above preprocessing, for any two vertices $q$ and $p$ of
$G$, $d(p,q)$ and the last edge of $\pi_q(p)$ can be obtained in constant time.

\vspace{-0.1in}
\paragraph{Remark.}
Another reason we compute shortest path trees using the algorithm in \cite{ref:MitchellAn89,ref:MitchellL192} instead of applying Dijkstra's algorithm on the graph $G$ is that a shortest path tree computed in $G$ may not be a planar tree. As will be seen later in Section~\ref{sec:equal}, our query algorithm will need to determine the relative positions of two shortest paths (from the same source), and to do so, we need shortest path trees that are planar.
\vspace{0.15in}

Given $s$ and $t$, following the discussion in
Section~\ref{sec:pre},
we assume that there are no trivial shortest \st\ paths and there is a
shortest \st\ path containing a gateway in $V^2_g(s,G)$ and a gateway in $V^2_g(t,G)$,
since otherwise the shortest path would have already been computed.
To simplify the notation, let $V(s)=V^2_g(s,G)$ and $V(t)=V^2_g(t,G)$.

A gateway of $V(s)$ is called a {\em via gateway} if there exists a shortest \st\ path that contains it.
Our goal is to find a via gateway, after which a shortest \st\ path
can be computed in additional $O(\log n)$ time by checking each gateway of $t$.
In the following, we present an $O(n_s+n_t\log n_s)$ time algorithm for finding a via gateway.
Without loss of generality, we assume that the first quadrant of $s$
has a via gateway. Below, we will describe our algorithm
only on the gateways of $V(s)$ in the first quadrant of $s$ (our
algorithm will run on each quadrant of $s$ separately). By slightly
abusing the notation, we still use $V(s)$ to denote the set of gateways of
$V(s)$ in the first quadrant of $s$.

Before describing our algorithm, we introduce some geometric
structures, among which the most important ones are a {\em gateway region} of $s$
and an {\em extended gateway region} of $t$. Chen et
al.~\cite{ref:ChenSh00} introduced the gateway region for
rectilinear polygonal domains and here we extend the concept to the
arbitrary polygonal domain case. In particular, our extended gateway
region has several new components that are critical to our algorithm, and it
may be interesting in its own right.

\subsection{The Gateway Region $R(s)$ for $s$}


Let $p_1,p_2,\ldots,p_k$ be the gateways of $s$ ordered from left to
right (e.g., see Fig.~\ref{fig:gatewaycutlines}). Note that each $p_i$ is a type-2 Steiner point on a projection
cut-line of $s$.
Let $l_1,l_2,\ldots,l_k$ be the projection cut-lines of $s$ that
contain these gateways, respectively, and thus they are also sorted
from left to right.
It is known~\cite{ref:ChenTw16,ref:ChenSh00} that
the $y$-coordinates
of $p_1,p_2,\ldots,p_k$ are in non-increasing order. The
sorted list can be obtained in $O(\log n)$ time when
computing $V(s)$~\cite{ref:ChenTw16,ref:ChenSh00},
and the list also follows the clockwise order around $s$.

For convenience of our discussion later, if $i$ is the smallest index
such that $y(p_{i})=y(p_{i+1})=\cdots=y(p_k)$, then we remove
$p_{i+1},\ldots,p_k$ from $V(s)$ because if there is a shortest \st\
path containing $p_j$ for any $j\in [i+1,k]$, then there must be a
shortest \st\ path containing $p_{i}$ as well. To simplify the
notation, we still use $k$ to denote the index of the last gateway of
$V(s)$ after the above removal procedure. Then we have the following
property: for any $i\in [1,k-1]$, $y(p_i)>y(p_k)$.

We define a {\em gateway region} $R(s)$ for $s$, as follows (e.g., see Fig.~\ref{fig:gatewayregion}).

Let $s_1$ be the intersection of $l_1$ with the horizontal line
through $p_k$.
For each $p_i$ with $i\in [2,k]$, project $p_i$ leftwards horizontally onto
$l_{i-1}$ at a point $p_{i}'$ (note that $p_i'=p_{i-1}$ if $y(p_{i-1})=y(p_i)$).
Define $R(s)$ as the region bounded by the line segments
connecting the points $s_1$, $p_1$, $p_2'$, $p_2$, \ldots, $p_k'$,
$p_k$, and $s_1$ in this cyclic order. Clearly, each edge of
$R(s)$ is either horizontal or vertical. Note that
$R(s)$ also includes the two segments $\overline{p_1p_2'}$ and $\overline{p_k'p_k}$.

\begin{figure}[t]
\begin{minipage}[t]{0.49\linewidth}
\begin{center}
\includegraphics[totalheight=1.5in]{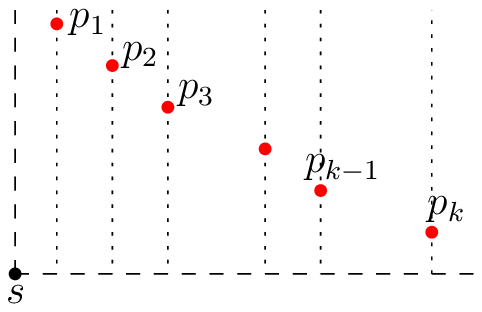}
\caption{\footnotesize
Illustrating the gateways of $V(s)$ and the cut-lines containing them.}
\label{fig:gatewaycutlines}
\end{center}
\end{minipage}
\hspace{0.05in}
\begin{minipage}[t]{0.49\linewidth}
\begin{center}
\includegraphics[totalheight=1.5in]{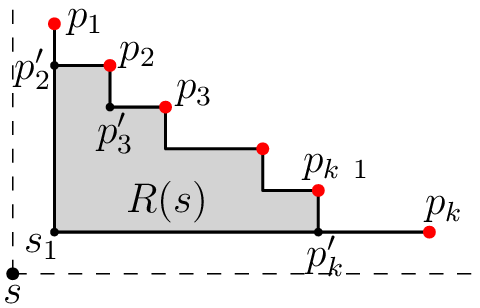}
\caption{\footnotesize
Illustrating the gateway region $R(s)$, which is the shaded region plus $\overline{p_1p_2'}$ and $\overline{p_kp_k'}$. The red points are gateways of $V(s)$.}
\label{fig:gatewayregion}
\end{center}
\end{minipage}
\vspace*{-0.15in}
\end{figure}

We use $\beta_s$ to denote the boundary portion of $R(s)$ from $p_1$ to
$p_k$ that contains all gateways of $V(s)$. We call $\beta_s$ the {\em
ceiling}, $\overline{s_1p_2'}$ the {\em left boundary}, and
$\overline{s_1p_k'}$ the {\em bottom boundary} of $R(s)$.
We refer to the region $R(s)$ excluding the points on $\beta_s$ as the {\em interior} of
$R(s)$.

\begin{observation}\label{obser:gatewayregion}
$R(s)$ is in $\calP$, and the interior of $R(s)$ does not contain any
polygon vertex of $\calP$.
\end{observation}
\begin{proof}
The lemma can be proved by similar techniques as
in~\cite{ref:ChenSh00} (e.g., Lemmas 3.7 and 3.8).
However, since the definition in~\cite{ref:ChenSh00} is particularly
for (weighted) rectilinear polygonal
domains, we present our own proof here, and this also makes our paper more
self-contained.

For each $i\in [2,k-1]$, define $w_i$ to be the intersection of the vertical line through $p_i$ and the horizontal line through $s$ (e.g., see Fig.~\ref{fig:gregion}).

\begin{figure}[h]
\begin{minipage}[t]{\linewidth}
\begin{center}
\includegraphics[totalheight=1.5in]{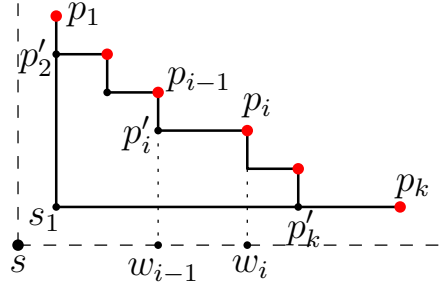}
\caption{\footnotesize
Illustrating the definition of $w_i$.}
\label{fig:gregion}
\end{center}
\end{minipage}
\vspace*{-0.15in}
\end{figure}

Consider the rectangle $R(w_{i-1},p_i)$ with $\overline{w_{i-1}p_i}$
as a diagonal. Since $p_{i-1}$ and $p_i$ are gateways of $V(s)$, both
of them are vertically visible to the horizontal line through $s$, and
thus, $\overline{p_{i-1}w_{i-1}}$ and $\overline{p_iw_i}$ are in
$\calP$. Also because $p_{i-1}$ and $p_i$ are gateways of $V(s)$,
neither $\overline{p_{i-1}w_{i-1}}\setminus{p_{i-1}}$ nor
$\overline{p_iw_i}\setminus{p_i}$ contains any polygon vertex.
Further, neither $\overline{p_{i-1}w_{i-1}}$ nor $\overline{p_iw_i}$
is contained a polygon edge since otherwise the edge would make $s$
not horizontally visible to $l_k$, i.e., the cut-line through $p_k$. Therefore, we obtain that $\overline{p_{i-1}w_{i-1}}\setminus{p_{i-1}}$ and $\overline{p_iw_i}\setminus{p_i}$ are in the interior of $\calP$. Since $s$ is horizontally visible to $l_k$, $\overline{w_{i-1}w_i}$ is in $\calP$.
Further, due to our general position assumption $s$ does not have the same $x$- or $y$-coordinate with any polygon vertex, $\overline{w_{i-1}w_i}$ is in the interior of $\calP$.

We claim that $R(w_{i-1},p_i)\setminus{\overline{p_i'p_i}}$ does not have a
polygon vertex that is vertically visible to $\overline{w_{i-1}w_i}$.
Assume to the contrary that this is not true, and let $p$ be the
lowest such vertex. Since
$\overline{p_{i-1}w_{i-1}}\cup\overline{w_{i-1}w_i}\cup
\overline{p_{i}w_{i}}$ is in $\calP$, $p$ must be horizontally visible
to $\overline{p_{i-1}w_{i-1}}$. Since $y(p)<y(p_i')\leq y(p_{i-1})$,
$p$ does not define a type-2 Steiner point at the cut-line $l_{i-1}$ since
otherwise $p_{i-1}$ would not be a gateway of $s$. Hence, there must be
a cut-line $l$ in $\calT$ between $p$ and $l_{i-1}$ such that $p$ defines
a type-2 Steiner point $p'$ on $l$ and $l$ is a proper ancestor of
$l_{i-1}$ (and thus prevents $p$ from defining a Steiner point on
$l_{i-1}$). Since $l$ is between $p$ and $l_{i-1}$, $s$ is horizontally visible to $l$. As $l_{i-1}$ is a projection cut-line of $s$ and $l$ is an ancestor of $l_{i-1}$, $l$ must also be a projection cut-line of $s$.
Further, by the definition of $p$, $p'$ is vertically visible to
the horizontal line through $s$. This implies that $V(s)$ must have a
gateway on $l$ no higher than $p'$, and thus the gateway is in
$R(w_{i-1},p_i)\setminus{\overline{p_i'p_i}}$, which incurs contradiction since
by definition $R(w_{i-1},p_i)\setminus\overline{p_i'p_i}$ does not have any
gateway of $s$.

The above claim, together with that $\overline{p_{i-1}w_{i-1}}\cup\overline{w_{i-1}w_i}\cup
\overline{p_{i}w_{i}}$ is in $\calP$, leads to that $R(w_{i-1},p_i)$
is in $\calP$. The observation can then be obtained due to the following: (1) $\overline{p_{i-1}w_{i-1}}\cup
\overline{p_{i}w_{i}}$ excluding $p_{i-1}$ and $p_i$ is in the interior of $\calP$, and (2) $R(s)$ is contained in the union of $R(w_{i-1},p_i)$ for all $i\in [2,k-1]$.
\qed
\end{proof}

\begin{figure}[t]
\begin{minipage}[t]{0.49\linewidth}
\begin{center}
\includegraphics[totalheight=1.0in]{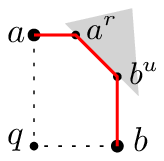}
\caption{\footnotesize
Illustrating the staircase path (the red solid) and the staircase region $R_s(a,b)$ (bounded by the solid path and the two dashed segments).}
\label{fig:defstair}
\end{center}
\end{minipage}
\hspace{0.05in}
\begin{minipage}[t]{0.49\linewidth}
\begin{center}
\includegraphics[totalheight=1.0in]{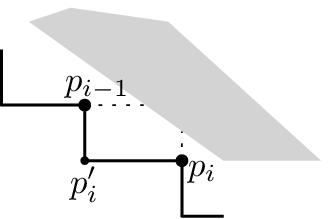}
\caption{\footnotesize
Illustrating the staircase region $R_s(p_{i-1},p_{i})$.}
\label{fig:pentagon}
\end{center}
\end{minipage}
\vspace*{-0.15in}
\end{figure}

For any two points $a$ and $b$ in the plane, we use $R(a,b)$ to denote
the rectangle with $\overline{ab}$ as a diagonal.
Suppose $a$ and $b$ of $\calP$ are both in the first quadrant of $s$ such that
$a$ is to the northwest of $b$. Recall that $a^r$ denotes the rightward projection of $a$ on $\partial\calP$ and $b^u$ denotes the upward projection of $b$ on $\partial\calP$.
With respect to $s$, we say that $a$
and $b$ are in {\em staircase positions} if either $\overline{aa^r}$ and
$\overline{bb^u}$ intersect, or both $a^r$ and $b^u$ are on the same
polygon edge (e.g., see Fig.~\ref{fig:defstair}); further, in the
former case, we call $\overline{ap}\cup \overline{pb}$
the {\em staircase path} between $a$ and $b$, where
$p=\overline{aa^r}\cap\overline{bb^u}$, and in the latter case, we
call  $\overline{aa^r}\cup\overline{a^rb^u}\cup \overline{b^ub}$ the
staircase path. The region bounded by the staircase path and
$\overline{aq}\cup \overline{qb}$, where $q$ is the intersection of
the vertical line through $a$ and the horizontal line through $b$, is
called the {\em staircase region} of $a$ and $b$ with respect to $s$, denoted by $R_s(a,b)$.
Roughly speaking, $R_s(a,b)$ is a pentagon after cutting the upper
right corner of $R(a,b)$ by a polygon edge.

\begin{observation}\label{obser:rectangle}
For each $i\in [2,k]$, e.g., see Fig.~\ref{fig:pentagon}, $p_{i-1}$ and $p_{i}$ are in staircase positions and the staircase region $R_s(p_{i-1},p_{i})$ is in $\calP$. Further, if $y(p_{i-1})> y(p_{i})$, then the interior of $R_s(p_{i-1},p_{i})$ along with its left and bottom edges $\overline{p_{i-1}p'_{i}}\cup\overline{p_{i}'p_{i}}\setminus\{p_{i-1},p_{i}\}$ does not contain a polygon vertex of $\calP$.
\end{observation}
\begin{proof}
The proof is somewhat similar to Observation~\ref{obser:gatewayregion}, so we only sketch it.
Recall that $y(p_{i-1})\geq y(p_{i})$.
If $y(p_{i-1})= y(p_{i})$, then
$p_{i-1}=p_{i}'$.
The proof of Observation~\ref{obser:gatewayregion} shows that
$R(w_{i-1},p_i)$ is in $\calP$. Since $\overline{p_{i}'p_{i}}$ is the
upper edge of $R(w_{i-1},p_i)$, $R_s(p_{i-1},p_{i})=\overline{p_{i}'p_i}$ is in $\calP$ and thus
$p_{i-1}$ and $p_{i}$ are in staircase positions.

In the following, we assume that $y(p_{i-1})> y(p_{i})$.
As in the proof of Observation~\ref{obser:gatewayregion}, $\overline{p_{i-1}p_{i}'}\setminus\{p_{i-1}\}$ does not contain any polygon vertex and is in the interior of $\calP$.
We claim that $R_s(p_{i-1},p_i)$ excluding the top edge and the right edge does not have a polygon vertex that is vertically visible to $\overline{p_{i}'p_i}$. The proof is similar to that in Observation~\ref{obser:gatewayregion}, and we omit the details. The claim, together with $\overline{p_{i-1}p_{i}'}\cup\overline{p_i'p_i}\in \calP$, leads to the observation.
\qed
\end{proof}

\begin{figure}[t]
\begin{minipage}[t]{0.48\linewidth}
\begin{center}
\includegraphics[totalheight=2.0in]{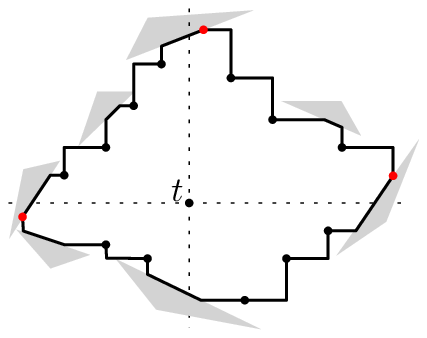}
\caption{\footnotesize
Illustrating the extended gateway region $R(t)$, bounded by the solid segments. The black points other than $t$ are all gateways and the three red points are special gateways to be defined later.}
\label{fig:extend}
\end{center}
\end{minipage}
\hspace{0.05in}
\begin{minipage}[t]{0.48\linewidth}
\begin{center}
\includegraphics[totalheight=2.0in]{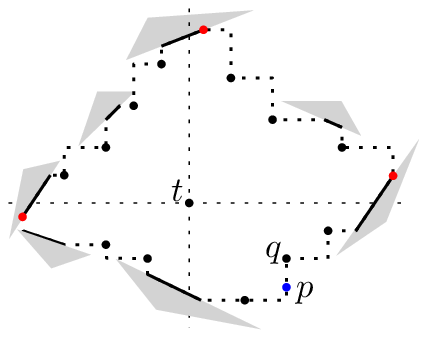}
\caption{\footnotesize
Illustrating the transparent edges on the boundary $R(t)$, shown by the dotted segments.
The three red points are special gateways. For illustrating Lemma~\ref{lem:extend}(4), a point $p$ on a transparent edge as well as an endpoint $q$ of the edge is also shown.
}
\label{fig:trans}
\end{center}
\end{minipage}
\vspace*{-0.15in}
\end{figure}

\subsection{The Extended Gateway Region $R(t)$ for $t$}

For $t$, we define an {\em extended gateway region} $R(t)$.
Unlike $R(s)$, which does not contain $s$, $R(t)$
contains $t$, e.g., see Fig.~\ref{fig:extend}.
Before giving the detailed definition of $R(t)$, which is
quite lengthy, we first discuss several key properties of it.

\paragraph{An overview of $R(t)$}
Let $\calV_1$ denote the set consisting of all polygon vertices and
their projection points on $\partial\calP$.
In general, $R(t)$ is a simple polygon that contains $t$. Let $\partial R(t)$ denote its boundary. Each edge of
$\partial R(t)$ is vertical, horizontal, or on a polygon edge. If an edge of $\partial R(t)$ is not on a polygon edge, then we call it a {\em transparent edge} (e.g., see
Fig.~\ref{fig:trans}).
It is the transparent edges that separate the interior of $R(t)$ from the
outside (i.e., for any point $p$ of $\calP$ outside $R(t)$, any path from $p$ to $t$ in $\calP$ must intersect a transparent edge of $R(t)$).
All gateways of
$V(t)$ are on  $\partial R(t)$. In addition, at most four points of
$\calV_1$ are considered as {\em special gateways} that are also on
$\partial R(t)$, and we include them in $V(t)$. Then, we have the following lemma (after removing some ``redundant'' gateways from $V(t)$).


\begin{lemma}\label{lem:extend}
\begin{enumerate}
\item
The point $t$ is visible to each gateway in $V(t)$.
\item
$R(t)$ is in $\calP$.

\item
For any point $p$ outside $R(t)$, there is a shortest path from $p$ to $t$ that contains a gateway in $V(t)$, and no shortest path from $p$ to $t$ contains more than one gateway of $V(t)$.

\item
For any point $p$ on a transparent edge $e$ of $R(t)$, one of the endpoints $q$
of $e$ is a gateway in $V(t)$ and $\overline{pq}\cup
\overline{qt}$ is an $xy$-monotone (and thus a shortest) path from $p$ to $t$ (e.g., see
Fig.~\ref{fig:trans}).

\item
For any point $p$ on a transparent edge of $R(t)$, if a shortest path $\pi(p,t)$ from $p$ to $t$ contains a gateway $q$ of $V(t)$, then $\overline{pq}$ is in $\pi(p,t)$ and is on a transparent edge $e$ of $R(t)$ (and $q$ is an endpoint of $e$).

\end{enumerate}
\end{lemma}

\paragraph{Remark.}
$R(s)$ and $R(t)$ are defined differently because $s$ and $t$ are not treated symmetrically in our algorithm. For example, we need $R(t)$ to have the properties in Lemma~\ref{lem:extend}, which are not necessary for $R(s)$. Also, as will be clear later, treating $s$ and $t$ differently helps us to further reduce the complexities of our data structure.
\vspace{0.08in}

In the sequel, we present define $R(t)$ in details, after which we will formally prove Lemma~\ref{lem:extend}.

Let $R_1(t)$ be the
sub-region of $R(t)$ in the first quadrant of $t$, which is defined as
follows (e.g. see Fig.~\ref{fig:enlarge}). The sub-regions of $R(t)$
in other quadrants are defined similarly.

\begin{figure}[t]
\begin{minipage}[t]{0.49\linewidth}
\begin{center}
\includegraphics[totalheight=1.7in]{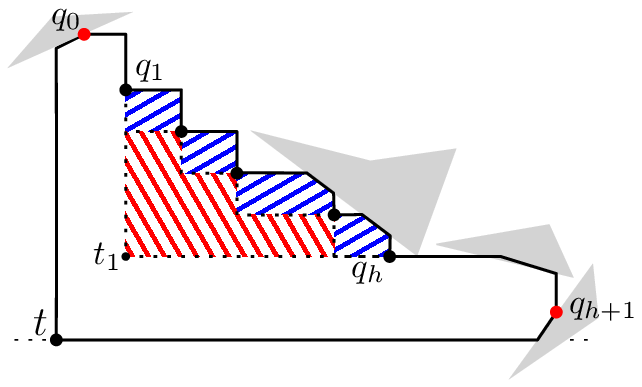}
\caption{\footnotesize
Illustrating $R_1(t)$, bounded by the solid segments. The red segments (of negative slope) illustrate $R'_1(t)$ while the blue segments (of positive slope) show the staircase regions, and their union is $R_1''(t)$.}
\label{fig:enlarge}
\end{center}
\end{minipage}
\hspace{0.05in}
\begin{minipage}[t]{0.49\linewidth}
\begin{center}
\includegraphics[totalheight=1.7in]{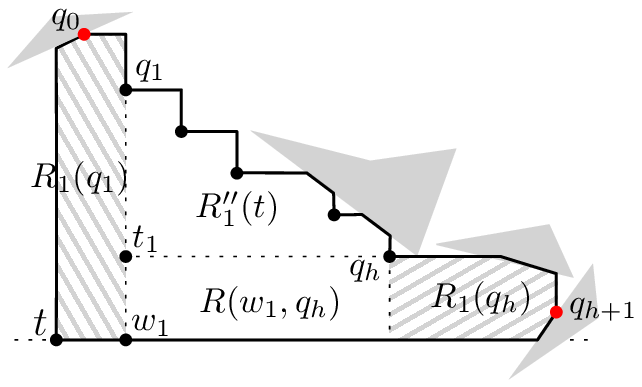}
\caption{\footnotesize
Illustrating the decomposition of $R_1(t)$ into four regions: $R'_1(t)$, $R(w_1,q_h)$, and two special regions $R_1(q_1)$ and $R_1(q_{h})$, to be defined later.
$q_0$ and $q_{h+1}$ are two points in $\calV_1$ to be defined later.
}
\label{fig:specialreg}
\end{center}
\end{minipage}
\vspace*{-0.15in}
\end{figure}

Let $R_1'(t)$ denote the same gateway region as $R(s)$ for $s$.
Let the gateways of $t$ on the ceiling of $R_1'(t)$ from left to right be $q_1,q_2,\ldots,q_{h}$. Let $R_1''(t)$ denote the union of $R_1'(t)$ and the staircase regions $R_t(q_{i-1},q_{i})$ (with respect to $t$) for all $i\in [2,h]$  (e.g. see Fig.~\ref{fig:enlarge}). By Observations~\ref{obser:gatewayregion} and \ref{obser:rectangle}, $R_1''(t)$ is in $\calP$ and does not contain any polygon vertex except on the boundary portion between $q_1$ and $q_h$.
Let $w_1$ denote the intersection of the vertical line through $q_1$
and the horizontal line through $t$ (e.g., see Fig.~\ref{fig:specialreg}). The proof of Observation~\ref{obser:gatewayregion} actually shows that the rectangle $R(w_1,q_h)$  is in $\calP$ and does not contain contain any polygon vertex except $q_h$.

The region
$R_1(t)$ is the union of $R_1''(t)$, $R(w_1,q_h)$, and two additional
regions $R_1(q_1)$ and $R_1(q_h)$, to be defined in the following
(e.g., see Fig.~\ref{fig:specialreg}). In order to define $R_1(q_1)$
and $R_1(q_h)$, we will also need to define two special points $q_0$ and $q_{h+1}$ from $\calV_1$.

\subsubsection{The region $R_1(q_1)$}

Let $l_h(t)$ and $l_v(t)$ be the horizontal and vertical lines through $t$, respectively.
For a sequence of points $a_1,a_2,\ldots,a_i$ in the plane, we use
$\square(a_1,a_2,\ldots,a_i)$ to denote the polygon with
$a_1,\ldots,a_i$ as vertices in this cyclic order on its boundary.

Let $e_u$ be the polygon edge that contains the upward projection $t^u$ of $t$.
Due to our general position assumption, $e_u$ is not horizontal.
Depending on whether the slope of $e_u$ is negative or positive, there are two cases for defining $R_1(q_1)$.

\begin{figure}[t]
\begin{minipage}[t]{0.49\linewidth}
\begin{center}
\includegraphics[totalheight=1.5in]{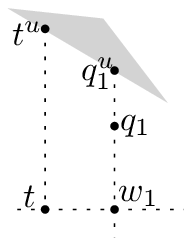}
\caption{\footnotesize
Illustrating the case where the slope of $e_u$ is negative.}
\label{fig:slopeneg}
\end{center}
\end{minipage}
\hspace{0.05in}
\begin{minipage}[t]{0.49\linewidth}
\begin{center}
\includegraphics[totalheight=1.7in]{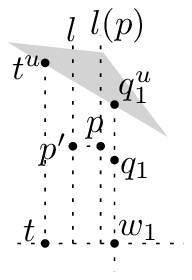}
\caption{\footnotesize
Illustrating the proof of Observation~\ref{obser:negative}.}
\label{fig:slopeneg1}
\end{center}
\end{minipage}
\vspace*{-0.15in}
\end{figure}

\begin{observation}\label{obser:negative}
If the slope of $e_u$ is negative  (e.g., see Fig.~\ref{fig:slopeneg}), then the upward projection $q_1^u$
of $q_1$ is on $e_u$. Further, the trapezoid $\square(t,w_1,q_1^u,t^u)$ is in $\calP$ and does not contain any polygon vertex except on $\overline{q_1q_1^u}$.

In this case, we define $R_1(q_1)$ as the above trapezoid.
\end{observation}
\begin{proof}
Since $q_1$ is a gateway, $\overline{w_1q_1}\setminus\{q_1\}$ is in $\calP$ and does not contain a polygon vertex.

We claim that no polygon vertex above $t$ and below $t^u$ is vertically visible to
$\overline{tw_1}\setminus\{w_1\}$. Indeed, assume to the contrary that this is not true. Then, let $p$ be the lowest such point (e.g., see Fig.~\ref{fig:slopeneg1}). Since the slope of $e_u$ is negative and $\overline{t^ut}\cup \overline{tw_1}$ is in $\calP$, $p$ must be horizontally visible to $\overline{tt^u}$.

By our definition of the graph $G$, there is a cut-line, denoted by $l(p)$, through $p$. Note that $l(p)$ is
between $t$ and the cut-line $l(q_1)$ through $q_1$, and $x(p)<x(q_1)$. Hence, $t$ is horizontally visible to $l(p)$.
Depending on whether $l(p)$ is a projection cut-line of $t$, there are two cases.

\begin{enumerate}
\item
If $l(p)$ is a projection cut-line of $t$, then since $p$ is type-2 Steiner point on $l(p)$ and $p$ is above $t$, $l(p)$ must have a gateway in $V_g^2(t,G)$ above $t$. But this contradicts with that $q_1$ is the leftmost gateway of $V_g^2(t,G)$ in the first quadrant of $t$.

\item
If $l(p)$ is not a projection cut-line of $t$, then there must be a cut-line $l$ in $\calT$ that is an ancestor of $l(p)$ such that $l$ is between $t$ and $l(p)$, i.e., $l$ prevents $l(p)$ from being a projection cut-line of $t$. We further let $l$ be such a cut-line in the highest node of $\calT$ (i.e., $l$ is still between $t$ and $l(p)$, and is an ancestor of $l(p)$). Then, $l$ must be a projection cut-line of $t$.

Since $p$ is horizontally visible to $\overline{tt^u}$, $p$ is also horizontally visible to $l$ and thus defines a type-2 Steiner point $p'$ on $l$ (e.g., see Fig.~\ref{fig:slopeneg1}). Clearly, $p'$ is vertically visible to the horizontal line $l_h(t)$. Therefore, $l$ also has a gateway of $V_g^2(t,G)$ in the first quadrant of $t$. But this contradicts with that $q_1$ is the leftmost gateway of $V_g^2(t,G)$ in the first quadrant of $t$.
\end{enumerate}

The claim is thus proved. The claim implies that $q_1^u$ is on $e_u$.
Further, due to the general position assumption, neither $\overline{t_ut}$ nor $\overline{tw_1}$ contains a polygon vertex. Recall that $\overline{w_1q_1}\setminus\{q_1\}$ does not have a polygon vertex. Hence,
the claim leads to the observation due to $\overline{t_ut}\cup \overline{tw_1}\cup \overline{w_1q_1^u}$ is in $\calP$.
\qed
\end{proof}


If the slope of $e_u$ is positive, we also need to define a point $q_0\in \calV_1$ on $t^u$.
Depending on whether $y(q_1)\geq y(t_u)$, there are two sub-cases.

\begin{figure}[t]
\begin{minipage}[t]{0.49\linewidth}
\begin{center}
\includegraphics[totalheight=1.3in]{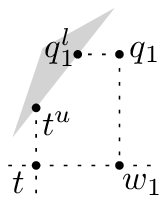}
\caption{\footnotesize
Illustrating the case where the slope of $e_u$ is positive and $y(q_1)\geq y(t^u)$.}
\label{fig:postive2}
\end{center}
\end{minipage}
\hspace{0.05in}
\begin{minipage}[t]{0.49\linewidth}
\begin{center}
\includegraphics[totalheight=1.3in]{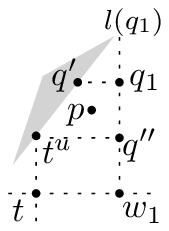}
\caption{\footnotesize
Illustrating the proof of Observation~\ref{obser:posabove}.}
\label{fig:posabove}
\end{center}
\end{minipage}
\vspace*{-0.15in}
\end{figure}

\begin{observation}\label{obser:posabove}
If the slope of $e_u$ is positive and $y(q_1)\geq y(t_u)$, e.g., see Fig.~\ref{fig:postive2}, then the
left projection $q_1^l$ of $q_1$ is on $e_u$, and the pentagon
$\square(t,w_1,q_1,q_1^l,t^u)$ is in $\calP$ and does not contain any polygon vertex
except on the top edge $\overline{q_1q_1^l}$.

In this case, we define $q_0$ as $q_1^l$, which must be in $\calV_1$, and define $R_1(q_1)$ as the above pentagon.
\end{observation}
\begin{proof}
Let $l(q_1)$ be the cut-line through $q_1$ and let $q'$ be the
intersection of the horizontal line through $q_1$ and the line
containing $e_u$. We will show that $q'=q_1^u$.
Let $q''$ be the intersection of $l(q_1)$ with the horizontal line
through $t^u$ (e.g. see Fig.~\ref{fig:posabove}). Note that $\overline{w_1q_1}\setminus\{q_1\}$ does not
contain any type-2 Steiner point.

First of all, by the similar proof as that for Observation~\ref{obser:negative}, we can show that no
polygon vertex above $t$ and below $t^u$ is vertically visible
to $\overline{tw_1}\setminus\{w_1\}$. This implies that the rectangle
$R(t,q'')$ does not contain any polygon vertex except $q_1$ (when $q_1=q''$), since
$\overline{t^ut}\cup\overline{tw_1}\cup\overline{w_1q_1}\setminus\{q_1\}$
is in $\calP$ and
does not contain a polygon vertex. This
further implies that $R(t,q'')$ is in $\calP$ and does not contain any
polygon vertex except possibly $q_1$. In the following, we focus on the
trapezoid $\square(t^u,q'',q_1,q')$, and we let $D$ denote the trapezoid but excluding the top edge
$\overline{q'q_1}$.

We claim that $D$ does not contain any polygon vertex. Assume to
the contrary that this is not true. Let $p$ be the lowest such vertex (e.g., see Fig.~\ref{fig:posabove}). Then
$y(p)<y(q_1)$, and $p$ is vertically visible to $\overline{tw_1}$ and
is horizontally visible to $\overline{q_1w_1}$.
Since $q_1$ is a gateway, $p$ does not define a Steiner point at
$l(p)$. This is only possible when there is a cut-line $l$ in
$\calT$ that is an ancestor of $l(q_1)$ and $l$ is between $p$ and
$l(q_1)$ (and $l\neq l(q_1)$). However, since $l$ is between $t$ and
$l(q_1)$ and $l$ is an ancestor of $l(q_1)$, $l$ would prevent $l(q_1)$ from
being a projection cut-line of $t$, incurring contradiction.

Since $D$ does not contain any polygon vertex and
$\overline{q't^u}\cup\overline{t^ut}\cup
\overline{tw_1}\cup \overline{w_1q_1}\subseteq \calP$, the above claim implies that $q'$ must be
$q_1^l$.

The above discussion also implies that the union of $\square(t^u,q'',q_1,q')$ and $R(t,q'')$, which is exactly the pentagon in the lemma statement, is in $\calP$ and does not contain any polygon vertex except the top edge $\overline{q_1q_1^l}$.
%
%

Finally, to see that $q_0$ must be a point in $\calV_1$, let $v$ be the polygon vertex defining the Steiner point $q_1$. Then, $q_0=q_1^l$ must be $v^l$, which is in $\calV_1$.
\qed
\end{proof}




\begin{figure}[t]
\begin{minipage}[t]{\linewidth}
\begin{center}
\includegraphics[totalheight=1.4in]{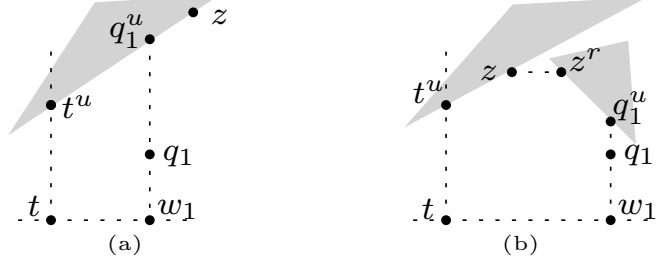}
\caption{\footnotesize
Illustrating the case where the slope of $e_u$ is positive and $y(q_1)< y(t^u)$: (a) $x(z)\geq x(q_1)$; (b) $x(z)<x(q_1)$.}
\label{fig:postive1}
\end{center}
\end{minipage}
\vspace*{-0.15in}
\end{figure}

\begin{observation}\label{obser:postive}
Suppose the slope of $e_u$ is positive and $y(q_1)< y(t_u)$. Let $z$ be the first point of
$\calV_1$ on $e_u$ to the right of $t_u$.
\begin{enumerate}
\item
\label{case:10}
If $x(z)\geq x(q_1)$, e.g., see Fig~\ref{fig:postive1}(a), then the upward projection $q_1^u$ of $q_1$ must be on $e_u$, and the trapezoid $\square (t,w_1,q_1^u,t^u)$ is in $\calP$ and does not contain any polygon vertex except on $\overline{q_1q_1^u}$.

In this case, $q_0$ is undefined and $R_1(q_1)$ is defined as the trapezoid $\square (t,w_1,q_1^u,t^u)$.
\item
If $x(z)<x(q_1)$, e.g., see Fig~\ref{fig:postive1}(b), then $z$ and $q_1$ are in staircase positions (with respect to $t$), and further, the region bounded by $\overline{zt^u}\cup \overline{t^ut}\cup \overline{tw_1}\cup\overline{w_1q_1}$ and the staircase path from $q_1$ to $z$ is
in $\calP$ and does not contain any polygon vertex  except on the horizontal edge (incident to $z$) and the vertical edge (incident to $q_1$) in the staircase path between $z$ and $q_1$.

In this case, we define $q_0$ as $z$ and define $R_1(q_1)$ as the region specified above.
\end{enumerate}
\end{observation}
\begin{proof}
By the similar proof as that for Observation~\ref{obser:negative}, we can show the following claim: No
polygon vertex above $t$ and below $t^u$ is vertically visible to $\overline{tw_1}\setminus\{w_1\}$.
We also claim that no polygon vertex  is horizontally visible to $\overline{t^uz}\setminus\{z\}$, since otherwise its leftward projection (which is in $\calV_1$) would be on $\overline{t^uz}\setminus \{z\}$, contradicting with the definition of $z$.

If $x(z)\geq x(q_1)$, then since $\overline{t^ut}\cup\overline{tw_1}\cup \overline{w_1q_1}$ is in $\calP$, the above two claims imply that
$q_1^u$ is on $e_u$. This further implies that $\square(t^u,t,w_1,q_1^u)$ is in $\calP$ and does not contain any polygon vertex except on $\overline{q_1q_1^u}$, since $\overline{q_1^ut^u}\cup \overline{t^ut}\cup\overline{tw_1}\cup \overline{w_1q_1}\setminus\{q_1,q_1^u\}$ does not contain a polygon vertex.

If $x(z)<x(q_1)$, since $\overline{zt^u}\cup \overline{t^ut}\cup\overline{tw_1}\cup \overline{w_1q_1}$ is in $\calP$, the above two claims imply that $z$ and $q_1$ are in staircase positions. As in the above case, this further implies that the region specified in the observation is in $\calP$ and does not contain any polygon vertex  except on the horizontal edge incident to $z$ and the vertical edge incident to $q_1$ in the staircase path between $z$ and $q_1$.
\qed
\end{proof}

\subsubsection{The region $R_1(q_h)$}

We proceed to define the region $R_1(q_h)$. Let $e_r$ be the polygon edge containing the right projection $t^r$ of $t$. Let $w_h$ be intersection of $l_h(t)$ and the vertical line through $q_h$. Depending on whether the slope of $e_r$ is negative or positive, there are two cases.

\begin{observation}\label{obser:qhneg}
If the slope of $e_r$ is negative, e.g., see Fig.~\ref{fig:qhneg}, then the right projection $q_h^r$ of $q_h$ is on $e_r$, and the trapezoid $\square(q_h,w_h,t^r,q_h^r)$ is in $\calP$ and does not contain any polygon vertex except on the top edge $\overline{q_hq_h^r}$.

In this case, we define $R_1(q_h)$ as the trapezoid $\square(q_h,w_h,t^r,q_h^r)$.
\end{observation}
\begin{proof}
We first claim that no polygon vertex above $w_h$ and strictly below
$q_{h}$ is vertically visible to $\overline{w_ht^r}$.
Indeed, assume to
the contrary this is not true. Let $p$ be the lowest such vertex. Note that $p$ cannot be on $\overline{q_hw_h}$ since otherwise $q_h$ would not be a gateway of $t$.
Let $l(q_{h})$ be the cut-line through $q_{h}$.
Since $\overline{q_{h}w_h}\cup \overline{w_ht^r}$ is in $\calP$, $p$ must be horizontally visible to $l(q_{h})$.
Due to $y(p)<y(q_{h})$ and $q_{h}$ is a gateway of $t$, $p$ cannot
define a type-2 Steiner point on $l(q_h)$. Hence, there is a cut-line
$l$ between $p$ and $l(q_{h})$ such that $l$ is an ancestor of
$l(q_{h})$ in $\calT$. We let $l$ be the highest such ancestor. Hence,
$p$ defines a type-2 Steiner point $p'$ at $l$. Since $l$ is between
$l(q_{h})$ and $t_r$, $t$ is horizontally visible to $l$. Since $l$ is
an ancestor of $l(q_{h})$ and $l(q_{h})$ is a projection cut-line of
$t$, $l$ must be a projection cut-line of $t$. Since $p'$ is a type-2
Steiner point vertically visible to $\overline{tt^r}$, $l$ also has a
gateway of $V_g^2(t,G)$ above $t$. But this contradicts with that
$q_{h}$ is the rightmost gateway of $V_g^2(t,G)$ in the first quadrant
of $t$.

As $\overline{q_hw_h}\cup \overline{w_ht^r}$ is in $\calP$, the above claim implies that the right projection $q_h^r$ of $q_h$ is on $e_r$, and the trapezoid $\square(q_h,w_h,t^r,q_h^r)$ is in $\calP$ and does not contain any polygon vertex except on the top edge $\overline{q_hq_h^r}$.
\qed
\end{proof}

\begin{figure}[t]
\begin{minipage}[t]{0.49\linewidth}
\begin{center}
\includegraphics[totalheight=1.0in]{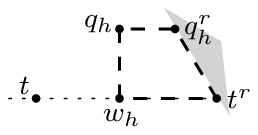}
\caption{\footnotesize
Illustrating the region $R_1(q_{h})$ (bounded by the thick dashed segments) in the case where the slope of $e_r$ is negative.}
\label{fig:qhneg}
\end{center}
\end{minipage}
\hspace{0.05in}
\begin{minipage}[t]{0.49\linewidth}
\begin{center}
\includegraphics[totalheight=1.0in]{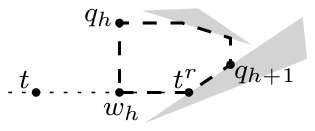}
\caption{\footnotesize
Illustrating the region $R_1(q_{h})$ (bounded by the thick dashed segments) in the case where the slope of $e_r$ is positive.}
\label{fig:rightside}
\end{center}
\end{minipage}
\vspace*{-0.15in}
\end{figure}



\begin{observation}\label{obser:qhpos}
If the slope of $e_r$ is positive, e.g., see Fig.~\ref{fig:rightside}, define $q_{h+1}$ to be the first point of $\calV_1$ on $e_r$ above $t^r$.
\begin{enumerate}
\item
$y(q_{h+1})\leq y(q_{h})$, and the two points $q_h$ and $q_{h+1}$ are in
staircase positions (with respect to $t$).
\item
The region bounded by $\overline{q_{h}w_h}\cup
\overline{w_ht^r}\cup \overline{t^rq_{h+1}}$ and the staircase path from
$q_{h+1}$ to $q_{h}$ is in $\calP$, and does not contain any polygon
vertex except on the vertical edge (incident to $q_{h+1}$) and the
horizontal edge (incident to $q_h$) in the staircase path from
$q_{h+1}$ to $q_h$.
\end{enumerate}
In this case, we define $R_1(q_h)$ as the region specified above.
\end{observation}
\begin{proof}
Since $q_h$ is type-2 Steiner point, its right projection on
$\partial\calP$ is in $\calV_1$. Based on this and due to $\overline{q_hw_h}\cup \overline{w_ht^r}\subseteq\calP$, we can show $y(q_{h+1})\leq y(q_{h})$. The analysis is similar as before and we omit the details.

By the same analysis as in the proof of Observation~\ref{obser:qhneg}, we can show that no polygon vertex above $w_h$ and strictly below
$q_{h}$ is vertically visible to $\overline{w_ht^r}$.

We claim that no polygon vertex above $t_r$ and strictly below
$q_{h}$ is vertically visible to $\overline{t^rq_{h+1}}\setminus\{q_{h+1}\}$.
Assume to the contrary this is not true. Let $p$ be such a vertex.
Then, the downward projection $p^d$ of $p$ is at
$\overline{t^rq_{h+1}}\setminus\{q_{h+1}\}$. But this contradicts with the
definition of $q_{h+1}$ since $p^d$ is in $\calV_1$.

The above two claims, together with $\overline{q_{h}w_h}\cup
\overline{w_ht^r}\cup \overline{t^rq_{h+1}}$ is in $\calP$, lead to that $q_h$ and $q_{h+1}$ are in
staircase positions and the region specified in the observation is in $\calP$.
Further, as discussed before, neither $\overline{q_hw_h}\setminus\{q_h\}$
nor $\overline{w_ht^r}\setminus\{t^r\}$ contains a polygon vertex of
$\calP$. This proves the observation.
\qed
\end{proof}


\subsubsection{A summary of the extended gateway region $R(t)$}

The above defined $R_1(q_1)$ and $R_1(q_h)$, and in some cases we also defined $q_0$ and $q_{h+1}$, both from $\calV_1$. We consider $q_0$ and $q_{h+1}$ as two special gateways for $t$ and include them in $V(t)$.
Note that both $q_0$ and $q_{h+1}$ can be computed in additional $O(\log n)$ time.

We perform the following {\em cleanup} procedure as part of our query
algorithm. If two consecutive gateways $q_i$ and $q_{i+1}$ for any
$i\in [0,h]$ have the same $x$-coordinate, then we remove $q_{i+1}$
from $V(t)$. The reason is that for any point $p\in \calP$ such that a
shortest path from $s$ to $p$ contains a $q_{i+1}$, there must be a
shortest path from $s$ to $p$ that contains $q_i$ because there is a
shortest path from $t$ to $q_{i+1}$ that contains $q_i$.
The cleanup procedure can be done in $O(n_t)$ time. Without loss of
generality, we assume that none of the gateways $q_0$ (if exists),
$q_1,q_2,\ldots,q_h$, $q_{h+1}$ (if exists)
has been removed by the cleanup procedure since otherwise we could
simply re-index them.
The following observation follows from our definition of $q_0$ and $q_{h+1}$ as well as the cleanup procedure.

\begin{observation}\label{obser:sorted}
The gateways $q_0$ (if exists), $q_1, \ldots, q_{h}$, and $q_{h+1}$ (if exists) are sorted by $x$-coordinate in strictly increasing order and also sorted by $y$-coordinate in strictly decreasing order.
\end{observation}

The definition of $R_1(t)$ is thus complete. So is the extended
gateway region $R(t)$, since  sub-regions of $R(t)$ in other quadrants of $t$
are defined similarly.
If we store the four projections on $\partial\calP$ for each
Steiner point of $G$ (this costs $O(n\log n)$ additional space),
then $R(t)$ can be explicitly computed in
$O(\log n)$ time.

Note that some edges of the boundary of $R(t)$ are on polygon edges, and we call other edges {\em transparent edges} (e.g., see Fig.~\ref{fig:trans}). We refer to the {\em outside} of $R(t)$
as the points of $\calP$ that are either not in $R(t)$ or on the transparent edges.
Clearly, for any point $p$ of $\calP$ outside $R(t)$, any path from $p$ to $t$ in $\calP$ must intersect a transparent edge of $R(t)$.

Lemma~\ref{lem:extend} given earlier summarizes some properties of $R(t)$ that will be
used later in our algorithm. We formally prove it below.

%
%
%
%
%

\paragraph{Proof of Lemma~\ref{lem:extend}.}
The first and second parts of the lemma can be seen from the
definition of $R(t)$ along with Observations~\ref{obser:gatewayregion},
\ref{obser:rectangle}, \ref{obser:negative}, \ref{obser:posabove},
\ref{obser:postive}, \ref{obser:qhneg}, and
\ref{obser:qhpos}.

For the third part, any shortest path $\pi(p,t)$ from $p$ to $t$ must intersect a point at a transparent edge of $R(t)$. Observe that each transparent edge is either horizontal or vertical. Further, for each transparent edge $e$, it always has an endpoint $v$ such that for each point $q\in e$, $\overline{vq}\cup\overline{qt}$ is a $xy$-monotone path from $q$ to $t$, and thus is a shortest path. Therefore, we obtain that there is a shortest path from $p$ to $t$ containing a gateway. On the other hand, assume to the contrary that the path contains two gateways $a$ and $b$ of $V(t)$. Without loss of generality, assume that we meet $a$ first if we move from $p$ to $t$ on the path, and thus $b$ is in the sub-path from $a$ to $t$. This implies that $b$ must be in the rectangle $R(t,a)$ since there is an $xy$-monotone path from $a$ to $t$. However, this is not possible according to our definition of $V(t)$ (in particular, due to the cleanup procedure).

The fourth part follows immediately from the above discussion.

For the fifth part, since there is an $xy$-monotone path from $t$ to $p$, $\pi(p,t)$ must be in the rectangle $R(p,t)$. Thus, $q$ is in $R(p,t)$. According to our definition of $R(t)$, $\overline{pq}$ must be on a transparent edge and $q$ is an endpoint of the edge. Since every transparent edge is either vertical or horizontal, $\overline{pq}$ is either horizontal or vertical, and thus $\overline{pq}$ is the only shortest path from $p$ to $q$. This implies that $\overline{pq}$ is in $\pi(p,t)$.
\qed

\begin{figure}[t]
\begin{minipage}[t]{\linewidth}
\begin{center}
\includegraphics[totalheight=1.2in]{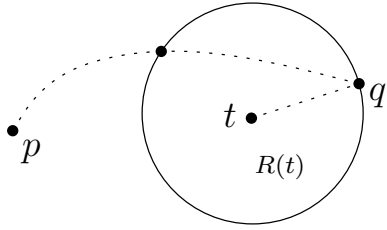}
\caption{\footnotesize
The following situation cannot occur: A shortest path  (the dotted
curve) from $t$ to a point $p$ outside $R(t)$ separates the boundary
of $R(t)$ (the solid circle) into two disjoint pieces. }
\label{fig:disconnect1}
\end{center}
\end{minipage}
\vspace*{-0.15in}
\end{figure}

\paragraph{Remark.} Lemma~\ref{lem:extend}(5) guarantees that for any
point $p$ outside $R(t)$, a shortest path $\pi(p,t)$ cannot separate
the boundary of $R(t)$ into two disjoint pieces (e.g., see Fig.~\ref{fig:disconnect1}).

\subsection{The Query Algorithm}

We have all necessary geometric prerequisites ready for explaining our algorithm.

Consider the gateway region $R(s)$ of $s$. Note that for any $p_i\in
V(s)$, there is always a shortest path from $s$ to $p_i$ containing
$s_1$ as there is an $xy$-monotone path from $s$ to $s_1$ in $\calP$.
Recall that we have assumed that there exists a shortest \st\ path that contains a gateway of $V(s)$.
The above implies that there exists a shortest path from $s_1$ to $t$
that contains a gateway of $V(s)$, and if we can find such a path,
by attaching an $xy$-monotone path from $s$ to $s_1$ to the path,
we can obtain a shortest \st\ path. For convenience, in the following,
we will focus on finding a shortest path from $s_1$ to
$t$ that contains a gateway of $V(s)$. By slightly abusing the notation, we still use $s$ to represent $s_1$. Again, our goal is to find a via gateway of $s$ in $V(s)$.

We first check whether there is a trivial shortest \st\ path in
$O(\log n)$ time. If yes, we are done. Otherwise, we proceed as follows. We begin with the following lemma.

\begin{lemma}\label{lem:disjoint}
If $R(t)$ contains a gateway $p$ of $V(s)$, then $\overline{sp}\cup\overline{pt}$ is a shortest \st\ path; otherwise, $R(s)$ does not intersect $R(t)$.
\end{lemma}
\begin{proof}
Assume that there is a point $p$ that is in both $R(s)$ and $R(t)$. In the following, we
first show that $t$ must be in the first quadrant of $p$.

By the definition of $R(s)$, the rectangle $R(s,p)$ is in
$R(s)$. For the rectangle $R(t,p)$, it may not be in $R(t)$, but this
only happens when one (or both) of its other two corners than $t$ and
$p$ is cut by a polygon edge (e.g., see Fig.~\ref{fig:tripath}).
In particular, we have the following observation: (1) If a point $q\in
R(t,p)$ is visible to $p$, then $q$ is also visible to $t$ and there
are two trivial shortest paths from $q$ to $t$ whose edges incident to
$q$ are vertical and horizontal, respectively (e.g., see
Fig.~\ref{fig:tripath}).

Note that $t$ cannot be in $R(s)$, since otherwise there would be a
trivial shortest \st\ path, a contradiction. Let $w_1$ and $w_2$ be the other two
corners of $R(p,t)$ such that $p,w_1,t,w_2$ are ordered clockwise on the boundary of $R(p,t)$.

Assume to the contrary that $t$ is not in the first quadrant of $p$.
Then, $t$ is in the second, third, or fourth quadrant of $p$. In the
following we will show that in each case there is a trivial
shortest \st\ path, which incurs contradiction.

If $t$ is in the second quadrant of $p$, then depending on whether $x(t)\geq x(s)$, there are two subcases.

\begin{figure}[t]
\begin{minipage}[t]{0.52\linewidth}
\begin{center}
\includegraphics[totalheight=0.8in]{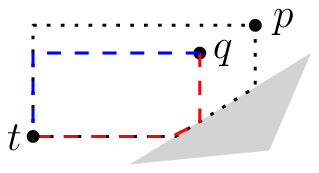}
\caption{\footnotesize
The region bounded by the dotted segments are the region of $R(t,p)$
contained in $R(t)$. The blue and red paths are two trivial shortest
paths from $q$ to $t$ whose edges incident to $q$ are vertical and
horizontal, respectively.}
\label{fig:tripath}
\end{center}
\end{minipage}
\hspace{0.05in}
\begin{minipage}[t]{0.47\linewidth}
\begin{center}
\includegraphics[totalheight=1.3in]{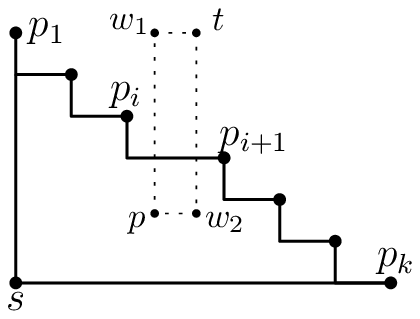}
\caption{\footnotesize
Illustrating the case where $t$ is in the first quadrant of $p$.}
\label{fig:quadrant1}
\end{center}
\end{minipage}
\vspace*{-0.15in}
\end{figure}

\begin{enumerate}
\item
If $x(t)\geq x(s)$, then $\overline{pw_1}$ must be in $R(s)$.
By the above observation, $t$ is visible to $w_1$ and thus is
vertically visible to the bottom boundary of $R(s)$. This implies that
there is a trivial shortest \st\ path.

\item
If $x(t)<x(s)$, then let $w$ be the intersection of $\overline{pw_1}$
and the left boundary of $R(s)$. By our above observation, there is a
trivial shortest path from $w$ to $t$ such that the edge of the path
incident to $w$ is vertical. Since $y(t)\geq y(w)\geq y(s)$, if we
append $\overline{sw}$ in front of the above path, we obtain a trivial
shortest \st\ path.
\end{enumerate}

If $t$ is in the third quadrant of $p$, then since $t\not\in R(s)$, $t$
cannot be in the first quadrant of $s$. Depending on which of the
other three quadrants of $s$ contains $t$, there are further three
subcases.

\begin{enumerate}
\item
If $t$ is in the second quadrant of $s$, then $w_1$ is in $R(s)$ and
thus $t$ is visible to $w_1$. Hence, $\overline{tw_1}$ intersects the
left boundary of $R(s)$, implying that there is a trivial shortest
\st\ path.

\item
If $t$ is in the third quadrant of $s$, then $s$ is in $R(t,p)$. Since
$s$ is visible to $p$, by our above observation, there is a trivial
shortest \st\ path.

\item

If $t$ is in the fourth quadrant of $s$, then $w_2$ is in $R(s)$ and
thus $t$ is visible to $w_2$. Hence, $\overline{tw_2}$ intersects the
bottom boundary of $R(s)$, implying that there is a trivial shortest
\st\ path.
\end{enumerate}

If $t$ is in the fourth quadrant of $p$, then
depending on whether $y(t)\geq y(s)$, there are two subcases.

\begin{enumerate}
\item
If $y(t)\geq y(s)$, then $\overline{pw_2}$ is in $R(s)$. By the above
observation, $t$ is visible to $w_2$ and is thus horizontally visible
to the left boundary of $R(s)$. Hence, there is a trivial shortest
\st\ path.

\item
If $y(t)<y(s)$, then $\overline{pw_2}$
intersects the bottom boundary of $R(s)$, say, at a point $w$. Since
$w$ is visible to $p$, by the above observation, there is a trivial
shortest path from $w$ to $t$ such that the edge of the path incident to $w$ is
horizontal. Since $x(t)\geq x(w)\geq x(s)$, if we append $\overline{sw}$ in front of the above path,
we obtain a trivial shortest \st\ path.
\end{enumerate}

The above proves that $t$ must be in the first quadrant of $p$. Since
$s$ is in the third quadrant of $p$, $\overline{sp}\cup \overline{pt}$
is a shortest \st\ path. This proves the lemma if $R(t)$ contains a gateway $p$ of $V(s)$.

In the following, we assume that $R(t)$ does not contain any gateway of $V(s)$. Our goal is to prove that $R(s)$ does not intersect $R(t)$. Assume to the contrary that this is not true and let $p$ be a point in  $R(s)\cap R(t)$. According to the above discussion, $t$ must be in the first quadrant of $p$. In the following, we show that there exists a trivial shortest \st\ path, which incurs contradiction.

Let $i\in [1,k]$ be the largest index such that $x(p_i)\leq x(p)$ (e.g., see Fig.~\ref{fig:quadrant1}).
Recall that the rightmost point of $R(s)$ is $p_k$. Since $R(t)$ does not contain any gateway of $V(s)$, $p$ is not $p_k$. This implies that $i< k$, and thus $p_{i+1}$ exists.
Depending on whether $x(t)< x(p_{i+1})$, there are two subcases.

\begin{enumerate}
\item
If $x(t)< x(p_{i+1})$, then $y(t)>y(p_{i+1})$ must hold since
otherwise $t$ would be in $R(s)$, e.g., see Fig.~\ref{fig:quadrant1}.
This implies that $\overline{pw_2}$ must be in $R(s)$. By the above
observation, $t$ is visible to $w_2$ and thus is vertically visible to
the bottom boundary of $R(s)$. This implies that there is a trivial shortest \st\ path.

\item
If $x(t)\geq x(p_{i+1})$, then $y(t)<y(p_{i+1})$ must hold since
otherwise $p_{i+1}$ would be in $R(t,p)$ and also in $R(t)$, contradicting with that $R(t)$ does not contain any gateway of $V(s)$. This implies that  $\overline{pw_1}$ must be in $R(s)$. By the above observation, $t$ is visible to $w_1$ and thus is horizontally visible to the left boundary of $R(s)$. This implies that there is a trivial shortest \st\ path.\qed
\end{enumerate}
\end{proof}

Our algorithm starts with checking whether $R(t)$
contains a gateway of $V(s)$. This can be done in $O(n_t+n_s)$ time, as
follows. We check the four quadrants of $t$ separately. Let
$R_1(t)$ be $R(t)$ in the first quadrant of $t$. To check whether
$R_1(t)$ contains a gateway of $V(s)$, we can simply scan the gateways of $V(s)$
and the gateways of $V(t)$ in $R_1(t)$ simultaneously from left to
right (somewhat like merge sort), which takes $O(n_t+n_s)$ time.
We do the same for other quadrants of $t$.

If $R(t)$ contains a gateway of $V(s)$, then by
Lemma~\ref{lem:disjoint},
we have found a shortest \st\ path. Otherwise, $R(s)$ and $R(t)$ are
disjoint and we proceed as follows.


By Lemma~\ref{lem:extend}, for each $p\in V(s)$, $d(p,t)=\min_{q\in V(t)}(d(p,q)+|\overline{qt}|)$, and we call such a gateway $q$ of $V(t)$ minimizing the above value a {\em coupled gateway} of $p$ and use $c(p)$ to denote it.

Our algorithm will compute a ``candidate'' coupled gateway $c'(p)$ for every gateway $p$ of $V(s)$ such that if $p\in V(s)$ is a via gateway, then $c(p)=c'(p)$.
Therefore, once the algorithm is done, the gateway $p$ that minimizes the value $|\overline{sp}|+d(p,c'(p))+|\overline{c'(p)t}|$ is a via gateway.

For any two points $a$ and $b$ on the ceiling $\beta_s$ of $R(s)$, we
use $\beta_s[a,b]$ to denote the sub-path of $\beta_s$ between
$a$ and $b$, which is $xy$-monotone.  This means that we can compute
$d(p_i,p_j)=\overline{|p_ip_j|}$ in constant time for every two
gateways $p_i$ and $p_j$ in $V(s)$.

We consider $V(t)$ as a cyclic list of points in {\em
counterclockwise} order around $t$ (we use ``counterclockwise''
since the list of $V(s)=\{p_1,p_2,\ldots,p_k\}$ are in
clockwise order around $s$).

We first compute $c(p_1)$ in a straightforward manner, i.e., check
every gateway of $V(t)$ (since $d(p,q)$ for any $p\in V(s)$ and $q\in V(t)$ is already computed in our preprocessing). This takes $O(n_t)$ time. We also compute
$c(p_k)$ in the same way. If there are multiple $c(p_k)$'s, then we
let $c(p_k)$ refer to the first one from $c(p_1)$ in the counterclockwise
order around $t$. Further, if there is more than one $c(p_1)$ from
the current $c(p_1)$ to $c(p_k)$ in the counterclockwise order, then we
update $c(p_1)$ to the one closest to $c(p_k)$. To simplify the
notation, let $q_1=c(p_1)$ and $q_k=c(p_k)$. Note that it is possible
that $q_1=q_k$.


The following lemma will be useful for circumventing the ``non-ideal'' situation depicted in Fig.~\ref{fig:dif1}.
Its correctness relies on the fact that the ceiling $\beta_s$ of $R(s)$ is $xy$-monotone (and thus is a shortest path).

\begin{lemma}\label{lem:40}
For any $p_i$ of $V(s)$, if $d(p_1,p_i)+d(p_i,q_1)=d(p_1,q_1)$, then $d(p_1,p_j)+d(p_j,q_1)=d(p_1,q_1)$ and $c(p_j)=q_1$ for each $j\in [1,i]$; similarly, if $d(p_k,p_i)+d(p_i,q_k)=d(p_k,q_k)$, then $d(p_k,p_j)+d(p_j,q_k)=d(p_k,q_k)$ and $c(p_j)=q_k$ for each $j\in [i,k]$.
\end{lemma}
\begin{proof}
We only prove the first part of the lemma since the second part is similar.

First of all, since $d(p_1,p_i)+d(p_i,q_1)=d(p_1,q_1)$, there is a shortest path from $p_1$ to $q_1$ that contains $p_i$. Because $\beta_s[p_1,p_i]$ is $xy$-monotone, there is a shortest path $\pi(p_1,q_1)$ from $p_1$ to $q_1$ that contains $\beta_s[p_1,p_i]$. Since $\beta_s[p_1,p_i]$ contains $p_j$, $\pi(p_1,q_1)$ contains $p_j$. Therefore, $d(p_1,p_j)+d(p_j,q_1)=d(p_1,q_1)$ holds.

Next, we prove that $c(p_j)=q_1$.
Assume to the contrary that there exists a point $q\in V(t)$ such that $d(p_j,q)+d(q,t)<d(p_j,q_1)+d(q_1,t)$.
Because $\beta_s[p_1,p_i]$ is $xy$-monotone and contains $p_j$, we have $d(p_1,p_i)=d(p_1,p_j)+d(p_j,p_i)$. Therefore, we can derive the following
\begin{equation}
\begin{split}
d(p_1,q)+d(q,t) &\leq d(p_1,p_j) + d(p_j,q) + d(q,t)\\
& < d(p_1,p_j) + d(p_j,q_1)+d(q_1,t) \\
& \leq d(p_1,p_j) + d(p_j,p_i) + d(p_i,q_1)+d(q_1,t) \\
& = d(p_1,q_1) + d(q_1,t).
\end{split}
\end{equation}

But this contradicts with that $q_1$ is a coupled gateway of $p_1$.
\qed
\end{proof}

%

Let $a_1$ be the largest index $i\in [1,k]$ such that
$d(p_1,p_i)+d(p_i,q_1)=d(p_1,q_1)$, which can be computed in $O(a_1)$
time, as follows. Starting from $i=2$, we simply check whether
$d(p_1,p_i)+d(p_i,q_1)=d(p_1,q_1)$, which can be done in $O(1)$ time
since $d(p_1,p_i)=|\overline{p_1p_i}|$ can be computed in constant
time and $d(p_i,q_1)$ has been computed in the preprocessing. If yes, we proceed with $i+1$; otherwise, we stop the algorithm and set $a_1=i-1$. We call the above a {\em stair-walking procedure}. The correctness is due to Lemma~\ref{lem:40}.

Similarly, define $b_k$ to be the smallest index $i\in [1,k]$ such that $d(p_k,p_i)+d(p_i,q_k)=d(p_k,q_k)$. By a symmetric stair-walking procedure, we can compute $b_k$ as well.
By Lemma~\ref{lem:40}, for each $i\in [1,a_1]\cup[b_k,k]$, $c(p_i)$ is computed. Hence, if $a_1\geq b_k$, then $c(p)$ for each $p\in V(s)$ is computed and we can finish the algorithm. Otherwise, we proceed as follows.

Our analysis will repeatedly use the following simple observation.

\begin{observation}
Suppose $p$ and $q$ are two points in a path $\pi$ in $\calP$. If the length of the sub-path of $\pi$ between $p$ and $q$ is not equal to $d(p,q)$, then $\pi$ cannot be a shortest path.
\end{observation}

Recall that $\pi_{q_1}(p_{a_1})$ is the shortest path between $q_1$ and $p_{a_1}$ in the shortest path tree $T(q_1)$, and $\pi_{q_k}(p_{b_k})$ is the shortest path in $T(q_k)$. The following two lemmas present our strategy for dealing with the non-ideal situation in which $\pi_{q_1}(p_{a_1})$ (resp., $\pi_{q_k}(p_{b_k})$) goes through the interior of $R(s)$ (e.g., see Fig.~\ref{fig:gothrough}).

\begin{lemma}\label{lem:interior}
\begin{enumerate}
\item
The shortest path $\pi_{q_1}(p_{a_1})$ contains a point in the interior of $R(s)$ only if its last edge (i.e., the edge incident to $p_{a_1}$) intersects the bottom boundary of $R(s)$, in which case the intersection at the bottom boundary of $R(s)$ has $x$-coordinate in $[x(s),x(p_{a_1+1})]$.
\item
The shortest path $\pi_{q_k}(p_{b_k})$ contains a point in the interior of $R(s)$ only if its last edge (i.e., the edge incident to $p_{b_k}$) intersects the left boundary of $R(s)$, in which case the intersection at the left boundary of $R(s)$ has $y$-coordinate in $[y(s),y(p_{b_k-1})]$.
\end{enumerate}
\end{lemma}
\begin{proof}
Note that since $a_1<b_k\leq k$, $p_{a_1+1}$ exists in $V(s)$. So does $p_{b_k-1}$.
We only prove the first part of the lemma, and the second part can be
proved in a similar way. To simplify the notation, let $i=a_1$.
Let $e$ be the last edge of $\pi_{q_1}(p_{a_1})$.
We assume that $\pi_{q_1}(p_{i})$ contains a point $w$ in the interior of $R(s)$.

\begin{figure}[t]
\begin{minipage}[t]{0.49\linewidth}
\begin{center}
\includegraphics[totalheight=1.3in]{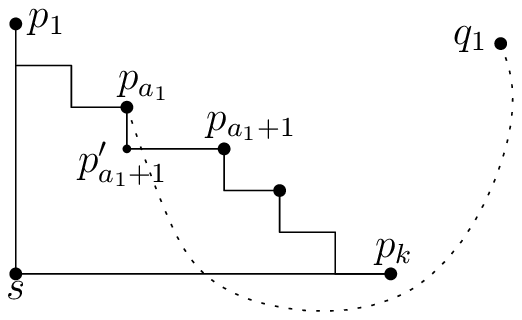}
\caption{\footnotesize The shortest path $\pi_{q_1}(p_{a_1})$ goes through the interior of $R(s)$.}
\label{fig:gothrough}
\end{center}
\end{minipage}
\hspace{0.05in}
\begin{minipage}[t]{0.49\linewidth}
\begin{center}
\includegraphics[totalheight=1.3in]{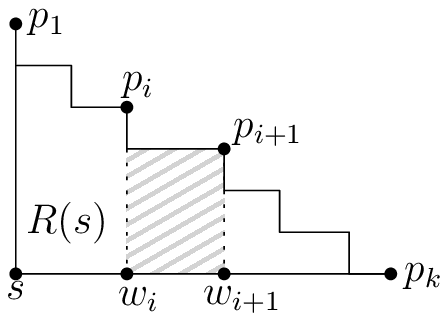}
\caption{\footnotesize The shaded region is $R(w_i,p_{i+1})$.}
\label{fig:recin}
\end{center}
\end{minipage}
\vspace*{-0.15in}
\end{figure}

Let $w_i$ (resp., $w_{i+1}$) be the intersection of the vertical line through $p_i$ (resp, $p_{i+1}$) with the bottom boundary of $R(s)$ (e.g., see Fig.~\ref{fig:recin}). Let $D=R(w_i,p_{i+1})$. We claim that $w$ must be in $D$. Indeed, assume to the contrary this is not true. Depending on whether $w$ is strictly to the left or right of $D$, there are two cases.

\begin{enumerate}
\item
If $w$ is strictly to the left of $D$, then $i>1$. By the definition
of $i=a_1$, $\beta_s[p_1,p_i]\cup \pi_{q_1}(p_i)$ is a shortest path
from $p_1$ to $q_1$, which contains $w$. Let $\pi$ represent the subpath of $\beta_s[p_1,p_i]\cup \pi_{q_1}(p_i)$ between $p_1$ and $w$. Note that $\pi$ contains $p_i$.
Since there is an $xy$-monotone path in $R(s)$ from $p_1$ to $w$, we have $d(p_1,w)=|\overline{p_1w}|$. On the other hand, since $w$ is strictly to the left of $D$, $p_i$ is not in $R(p_1,w)$. This implies that the length of $\pi$ must be larger than $d(p_1,w)=|\overline{p_1w}|$, contradicting with that $\beta_s[p_1,p_i]\cup \pi_{q_1}(p_i)$ is a shortest path from $p_1$ to $q_1$.

\item
If $w$ is strictly to the right of $D$, then there is an $xy$-monotone path from $p_i$ to $w$ that contains $p_{i+1}$ and the path is contained in a shortest path from $p_i$ to $q_1$. Hence,
$d(p_i,q_1)=d(p_i,p_{i+1})+d(p_{i+1},q_1)$. Since
$d(p_1,q_1)=d(p_1,p_i)+d(p_i,q_1)$, we obtain that
$d(p_1,q_1)=d(p_1,p_i)+d(p_i,p_{i+1})+d(p_{i+1},q_1)=d(p_1,p_{i+1})+d(p_{i+1},q_1)$.
But this contradicts with the definition of $i=a_1$.

\end{enumerate}

The above proves that $w$ must be in $D$. Observe that $d(p_i,w)=|\overline{p_iw}|$. Further, due to  Observations~\ref{obser:gatewayregion} and \ref {obser:rectangle}, $p_i$ is visible to $w$.
By Observation~\ref{obser:gatewayregion} and due to $R(s)\cap R(t)=\emptyset$, the endpoint of $e$ other
than $p_i$, which is a polygon vertex or $q_1$, is not in $R(s)$.
Hence, $w$ must be contained in $e$ and $e$ must intersect
the bottom boundary of $R(s)$. Further, according to the above claim,
every point $w\in e\cap R(s)$ must be in $D$, and thus, the
intersection of $e$ and the bottom boundary of $R(s)$ must be in $D$.
This proves the lemma. \qed
\end{proof}


\begin{lemma}\label{lem:70}
If the last edge of $\pi_{q_1}(p_{a_1})$ intersects the bottom boundary of $R(s)$, or
the last edge of $\pi_{q_1}(p_{b_k})$ intersects the left boundary of $R(s)$, then
$p_j$ for each $j\in [a_1+1,b_k-1]$ cannot be a via gateway.
\end{lemma}
\begin{proof}
We only prove the case for $\pi_{q_1}(p_{a_1})$ since the other
case is similar. To simplify the notation, let $i=a_1$.
Let $e$ be the last edge of  $\pi_{q_1}(p_{a_1})$, which intersects
the bottom boundary of $R(s)$, say, at a point $w$  (e.g., see Fig.~\ref{fig:notpossible2}). By Lemma~\ref{lem:interior}, $x(w)\in [x(s),x(p_{i+1})]$.
Consider any $j\in [a_1+1,b_k-1]$.
In the following, we show that $p_j$ cannot be a via gateway.
Since $j<b_k\leq k$, $y(p_j)> y(p_k)=y(s)=y(w)$.

\begin{figure}[t]
\begin{minipage}[t]{0.49\linewidth}
\begin{center}
\includegraphics[totalheight=1.3in]{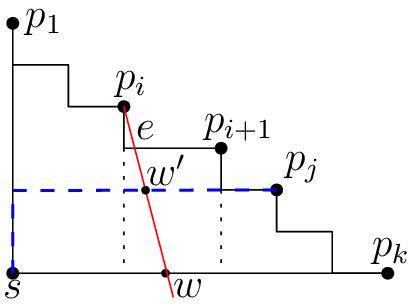}
\caption{\footnotesize
Illustrating the proof of Lemma~\ref{lem:70}. The red edge $e$ is the last edge in $\pi_{q_1}(p_i)$. The blue dashed path is the subpath of $\pi(s,t)$ between $s$ and $p_j$.}
\label{fig:notpossible2}
\end{center}
\end{minipage}
\hspace{0.05in}
\begin{minipage}[t]{0.49\linewidth}
\begin{center}
\includegraphics[totalheight=1.2in]{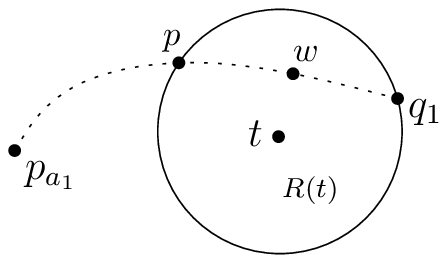}
\caption{\footnotesize
This following situation cannot occur: The path $\pi_{q_1}(p_{a_1})$ (the dotted curve) separates the boundary of $R(t)$ (the solid circle) into two (or more) disconnected pieces. }
\label{fig:disconnect}
\end{center}
\end{minipage}
\vspace*{-0.15in}
\end{figure}

Assume to the contrary that $p_j$ is a via gateway. Then there is a shortest \st\ path $\pi(s,t)$ that contains $p_j$. Without loss of generality, we assume that the sub-path of $\pi(s,t)$ between $s$ and $p_j$, denoted by $\pi(s,p_j)$, consists of a vertical segment through $s$ and a horizontal segment through $p_j$ (e.g., see Fig.~\ref{fig:notpossible2}).
Since $j>i$ and $x(w)\in [x(s),x(p_{i+1})]$, $e$ intersects the horizontal segment of $\pi(s,p_j)$ at a point $w'$ and $x(w')\in [x(s),x(p_{i+1})]$. Note that $y(w')>y(w)=y(s)$ since $y(p_j)>y(s)$. As $c(p_i)=q_1$,  $\pi_{q_1}(p_i)\cup\overline{q_1t}$ is a shortest path from $p_i$ to $t$. Hence, the sub-path of $\pi_{q_1}(p_i)\cup\overline{q_1t}$ between $w'$ and $t$ is a shortest path from $w'$ to $t$. Also, as $w'\in \pi(s,t)$, the sub-path of $\pi(s,t)$ between $w'$ and  $t$ is also a shortest path from $w'$ to $t$.
Therefore, the concatenation of $\pi_1$ and $\pi_2$, denoted by $\pi$, is also a shortest \st\ path, where $\pi_1$ is the sub-path of $\pi(s,p_j)$ between $s$ and $w'$ and $\pi_2$ is the sub-path of $\pi_{q_1}(p_i)\cup\overline{q_1t}$  between $w'$ and $t$.
Notice that $\pi$ contains $s$, $w'$, and $w$ in this order. Since $y(w')>y(s)=y(w)$, the length of the subpath between $s$ and $w$ is strictly larger than $d(s,w)=|\overline{sw}|$. However, this contradicts with that $\pi$ is a shortest \st\ path.
Hence, $p_j$ cannot be a via gateway. The lemma thus follows.
\qed
\end{proof}

Due to our preprocessing, we check in constant time whether the last edge of $\pi_{q_1}(p_{a_1})$ intersects the bottom boundary of $R(s)$. Similarly, we can check whether the last edge of $\pi_{q_1}(p_{b_k})$ intersects the left boundary of $R(s)$.
If the answer is yes for either case, then by Lemma~\ref{lem:70}, we can stop
the algorithm (i.e., no need to compute the coupled gateways for any $p_i$ with $i\in [a_1+1,b_k-1]$). Otherwise, by Lemma~\ref{lem:interior}, neither
$\pi_{q_1}(p_{a_1})$ nor $\pi_{q_1}(p_{b_k})$ contains a point
in the interior of $R(s)$. Thus, the situation depicted in
Fig.~\ref{fig:gothrough} does not happen to either path. Our algorithm
proceeds as follows.

Due to the properties of  $R(t)$ in Lemma~\ref{lem:extend}, the following lemma shows that $\pi_{q_1}(p_{a_1})$ (resp., $\pi_{q_k}(p_{b_k})$) cannot separate the boundary of $R(t)$ into two disconnected pieces (e.g., see Fig.~\ref{fig:disconnect}).
\begin{lemma}\label{lem:disconnect}
The path $\pi_{q_1}(p_{a_1})$ (resp., $\pi_{q_k}(p_{b_k})$)
does not contain any point in the interior of $R(t)$, and thus, the
intersection of the path with $\partial R(t)$ is connected.
Further, $q_1$ is the only gateway of $V(t)$ in
$\pi_{q_1}(p_{a_1})$,
and similarly, $q_k$ is the only gateway of $V(t)$ in
$\pi_{q_k}(p_{b_k})$.
\end{lemma}
\begin{proof}
We only discuss the case for $\pi_{q_1}(p_{a_1})$, since the case for the other path is similar.

Assume to the contrary that $\pi_{q_1}(p_{a_1})$ contains a point
$w$ in the interior of $R(t)$. Then, the subpath from $p_{a_1}$ to $w$
must intersect a transparent edge of $R(t)$ at a point $p$ (e.g., see
Fig.~\ref{fig:disconnect}). Let $\pi=\pi_{q_1}(p_{a_1})\cup
\overline{q_1t}$. Since $\pi$ is a shortest path from $p_{a_1}$ to
$t$, $q_1$ must be in the rectangle $R(t,p)$. By Lemma~\ref{lem:extend}(5), the subpath of $\pi$ from $p$ to
$q_1$ must be the line segment $\overline{q_1t}$, which is on $\partial R(t)$. However, this contradicts with that the subpath of
$\pi$ from $p$ to $q_1$ contains a point $w$ in the interior of
$R(t)$.
Further, since $\pi$ is a shortest path, by Lemma~\ref{lem:extend}(3), $\pi$ only
contains a single gateway of $V(t)$. Hence, $q_1$ is the only gateway of $V(t)$ in
$\pi_{q_1}(p_{a_1})$.
\qed
\end{proof}

Recall that  $q_1=q_k$ is possible.
Depending on whether $q_1=q_k$, there are two cases. In the following,
we first describe our algorithm for the {\em unequal case} $q_1\neq q_k$, and later
we will show that the {\em equal-case} $q_1= q_k$ can be reduced to the unequal case.

\subsubsection{The unequal case $q_1\neq q_k$}

Since $q_1\neq q_k$, $q_1$ and $q_k$ partition the cyclic list $V(t)$ into two sequential lists, one of which has $q_1$ as the first point and $q_k$ as the last point following the counterclockwise order around $t$, and we use $V_t(1,k)$ to denote that list. The following observation follows from our definitions of $q_1$ and $q_k$.

\begin{observation}\label{obser:tie}
Suppose $q$ is a gateway in $V_t(1,k)$.
\begin{enumerate}
\item
If $q\neq q_1$, then $d(p_1,q_1)+d(q_1,t)<d(p_1,q)+d(q,t)$, which further implies that $d(p_{a_1},q_1)+d(q_1,t)<d(p_{a_1},q)+d(q,t)$.
\item
if $q\neq q_k$, then $d(p_k,q_k)+d(q_k,t)<d(p_k,q)+d(q,t)$, which further implies that $d(p_{b_k},q_k)+d(q_k,t)<d(p_{b_k},q)+d(q,t)$.
\end{enumerate}
\end{observation}
\begin{proof}
We only prove the first part of the observation, since the second part is similar.
By the definitions of $q_1$ and $q_k$, we can immediately obtain that $d(p_{1},q_1)+d(q_1,t)< d(p_{1},q)+d(q,t)$. Further, by the definition of $a_1$, $d(p_{1},q_1)=d(p_1,p_{a_1})+d(p_{a_1},q_1)$. On the other hand, it holds that $d(p_{1},q)\leq d(p_1,p_{a_1})+d(p_{a_1},q)$. The above three inequalities together lead to $d(p_{a_1},q_1)+d(q_1,t)< d(p_{a_1},q)+d(q,t)$. \qed
\end{proof}

For any $i$ and $j$ with $1\leq i\leq j\leq k$, we use the interval $[i,j]$ to represent the gateways $p_i,p_{i+1},\ldots,p_j$.
Our algorithm works on the interval $[1,k]$ and $V_t(1,k)$.
Since $q_1\neq q_k$, we have the following lemma.

\begin{lemma}\label{lem:intersect}
The two paths $\pi_{q_1}(p_{a_1})$ and $\pi_{q_k}(p_{b_k})$ do not intersect.
\end{lemma}
\begin{proof}
Since $q_1\neq q_k$, by Observation~\ref{obser:tie}, $d(p_{a_1},q_1)+d(q_1,t)< d(p_{a_1},q_k)+d(q_k,t)$.
Assume to the contrary that $\pi_{q_1}(p_{a_1})$ and $\pi_{q_k}(p_{b_k})$ intersect, say, at a point $w$ (e.g., see Fig.~\ref{fig:intersect}).

Let $\pi_1$ be the path $\pi_{q_1}(p_{a_1})\cup \overline{q_1t}$ and let $\pi_2$ be the path
$\pi_{q_k}(p_{b_k})\cup \overline{q_kt}$.
Let $\pi_1'$ be the sub-path of $\pi_1$ between $w$ and $t$. Let $\pi_2'$ be the sub-path of $\pi_2$ between $w$ and $t$.

If we replace $\pi_1'$ by $\pi_2'$ in $\pi_1$, we obtain a path $\pi_3$ from $p_{a_1}$ to $t$ that contains $q_k$, and the length of $\pi_3$ is at least $d(p_{a_1},q_k)+d(q_k,t)$. Since $d(p_{a_1},q_1)+d(q_1,t)< d(p_{a_1},q_k)+d(q_k,t)$, the length of $\pi_3$ is larger than that of $\pi_1$. This further implies $|\pi_1'|$ (i.e., the length of $\pi_1'$) is smaller than $|\pi_2'|$.

Now if we replace $\pi_2'$ by $\pi_1'$ in $\pi_2$, then we obtain another path $\pi_4$ from $p_{b_k}$ to $t$ that contains $q_1$. Since $|\pi_1'|<|\pi_2'|$, we obtain that $|\pi_4|<|\pi_2|$. As $d(p_{b_k},q_1)+d(q_1,t)\leq |\pi_4|$, we have $d(p_{b_k},q_1)+d(q_1,t)< |\pi_2|=d(p_{b_k},q_k)+d(q_k,t)$. However, this contradicts with the definition of $q_k$.
\qed
\end{proof}

\begin{figure}[t]
\begin{minipage}[t]{0.49\linewidth}
\begin{center}
\includegraphics[totalheight=1.7in]{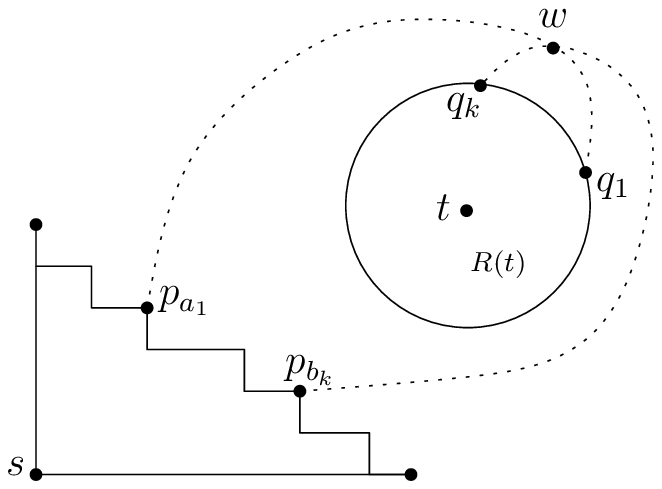}
\caption{\footnotesize
Illustrating the proof of Lemma~\ref{lem:intersect}. The two paths $\pi_{q_1}(p_{a_1})$ and $\pi_{q_k}(p_{b_k})$ intersect at $w$.}
\label{fig:intersect}
\end{center}
\end{minipage}
\hspace{0.05in}
\begin{minipage}[t]{0.49\linewidth}
\begin{center}
\includegraphics[totalheight=1.7in]{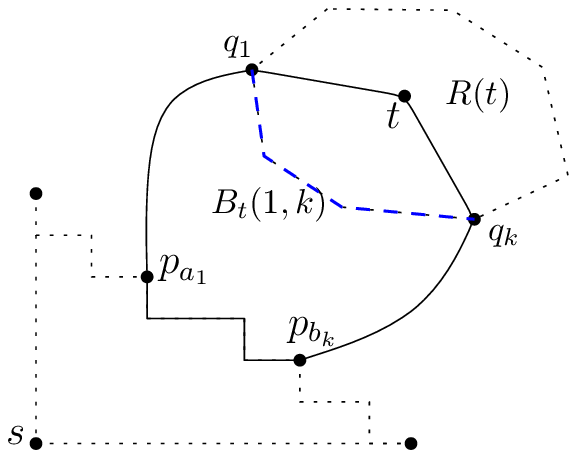}
\caption{\footnotesize
The region bounded by the sold curves is $D'$. The (blue) bold dashed boundary portion of $\partial R(t)$ is $B_t(1,k)$.}
\label{fig:cycle}
\end{center}
\end{minipage}
\vspace*{-0.15in}
\end{figure}

Lemma~\ref{lem:enclose} shows why we need the list $V_t(1,k)$, and its proof will need Lemma~\ref{lem:optint}, which shows an important property of a shortest \st\ path.

\begin{lemma}\label{lem:optint}
Suppose $\pi(s,t)$ is a shortest path that contains a gateway $p\in V(s)$. Then, the sub-path of $\pi(s,t)$ between $p$ and $t$ does not contain any interior point of $R(s)$.
\end{lemma}
\begin{proof}
Let $\pi(p,t)$ be the subpath of $\pi(s,t)$ from $p$ to $t$. Assume to the contrary that $\pi(p,t)$ contains a point $w$ in the interior of $R(s)$. Then, by Observation~\ref{obser:gatewayregion}, $d(s,w)=|\overline{sw}|$. Therefore, the length of the subpath of $\pi(s,t)$ between $s$ and $w$, which contains $p$, is equal to $|\overline{sw}|$. This is possible only if $p$ is in the rectangle $R(s,w)$. Since $w$ is in the interior of $R(s)$, all points of $R(s,w)$ are in the interior of $R(s)$. However, by definition, the gateway $p$, which is on the ceiling of $R(s)$, is not in the interior of $R(s)$. Therefore, $p$ cannot be in $R(s,w)$. This incurs contradiction.
\qed
\end{proof}

\begin{lemma}\label{lem:enclose}
For any gateway $p_j$ with $j\in [a_1+1,b_k-1]$, if $p_j$ is a via gateway, then it has a coupled gateway in $V_t(1,k)$.
\end{lemma}
\begin{proof}
Suppose $p_j$ is a via gateway  with $j\in [a_1+1,b_k-1]$. Thus, there is a shortest \st\ path that contains $p_j$,
and we let $\pi(p_j,t)$ denote the sub-path from $p_j$ to $t$.

If $\pi(p_j,t)$ intersects the path $\pi_{q_1}(p_{a_1})$, say, at a point $w$, then we claim that $q_1$ is coupled gateway of $p_j$. Indeed, observe that $q_1$ is a gateway $q$ in $V(t)$ that minimizes the value $d(w,q)+d(q,t)$. Since $\pi(p_j,t)$ contains $w$, $q_1$ is also a gateway $q$ in $V(t)$ that minimizes the value $d(p_j,q)+d(q,t)$. Thus, $q_1$ is a coupled gateway of $p_j$. As $q_1\in V_t(1,k)$, the lemma holds.

If $\pi(p_j,t)$ intersects the path $\pi_{q_k}(p_{b_k})$, then by the similar analysis as above, $q_k$ is a coupled gateway of $p_j$.  As $q_k\in V_t(1,k)$, the lemma also holds for this case.

In the following, we assume that  $\pi(p_j,t)$ does not intersect either $\pi_{q_1}(p_{a_1})$ or $\pi_{q_k}(p_{b_k})$.

By Lemma~\ref{lem:intersect}, $\pi_{q_1}(p_{a_1})$ and $\pi_{q_k}(p_{b_k})$ do not intersect. Recall that neither path contains an interior point of $R(s)$. Hence, $\pi_{q_1}(p_{a_1})$, $\pi_{q_k}(p_{b_k})$, $\overline{q_1t}$, $\overline{q_kt}$, and $\beta_s[p_{a_1},p_{b_k}]$ together form a closed curve that divides the plane into two regions (e.g., see Fig.~\ref{fig:cycle}), one of which (denoted by $D'$) does not contain $s$. Let $B_t(1,k)$ be the boundary portion of $R(t)$ contained in $D'$. Lemma~\ref{lem:disconnect} implies that the set of gateways of $V(t)$ on $B_t(1,k)$ is exactly $V_t(1,k)$. Further, $B_t(1,k)$ divides $D'$ into two subregions: one of them, denoted by $D(1,k)$, contains $\beta_s[p_{a_1},p_{b_k}]$ (and thus contains $p_j$), and the other contains $t$ (e.g., see Fig.~\ref{fig:cycle}).

By Lemma~\ref{lem:optint}, $\pi(p_j,t)$ does not contain any point in the interior of $R(s)$.
Since $p_j$ is in $D(1,k)$ and $t$ is not, and $\pi(p_j,t)$ does not intersect either $\pi(p_{a_1},q_1)$ or $\pi(p_{b_k},q_k)$, $\pi(p_j,t)$ must intersect $B_t(1,k)$. By Lemma~\ref{lem:extend} and our definition of $B_t(1,k)$, for any point $p$ in $B_t(1,k)$, $B_t(1,k)$ contains a gateway $q$ such that there is an $xy$-monotone path from $p$ to $t$ that contains $q$, and further, $q$ is in $V_t(1,k)$ since $V(t)\cap B_t(1,k)=V_t(1,k)$. Consequently, since $\pi(p_j,t)$ intersects $B_t(1,k)$, we obtain that $V_t(1,k)$ has a gateway $q$ such that there is a shortest path from $p_j$ to $t$ that contains $q$.  This leads to the lemma.
\qed
\end{proof}

In light of Lemma~\ref{lem:enclose}, to compute the candidate coupled gateways for all $p_i$ with $i\in [a_1+1,b_k-1]$, we only need to consider the gateways in $V_t(1,k)$. In the following, we work on the problem recursively. We may consider each recursive step as working on a subproblem, denoted by $([i',j'], [i,j],V_t(i,j))$ with $[i',j']\subseteq [i,j]\subseteq [1,k]$, where the goal is to find candidate coupled  gateways from a sublist $V_t(i,j)$ of $V_t(1,k)$ for the gateways in $[i',j']$, and further, there exist a shortest path from $p_{a_i}$ to the first point of $V_t(i,j)$ and a shortest path from $p_{b_j}$ to the last point of $V_t(i,j)$ such that the two paths do not intersect and neither path contains a point in the interior of $R(s)$. Initially, our subproblem is $([a_1+1,b_k-1],[1,k],V_t(1,k))$. We proceed as follows.

If $b_k-1=a_1+1$, then the interval $[a_1+1,b_k-1]$ has only one gateway $p$. We simply check all gateways of $V_t(1,k)$ to find the point $q$ that minimizes the value $d(p,q)+d(q,t)$ among all $q\in V_t(1,k)$, and then return $q$ as the candidate coupled gateway of $p$. The algorithm can stop. Otherwise, we proceed as follows.

\begin{figure}[t]
\begin{minipage}[t]{\linewidth}
\begin{center}
\includegraphics[totalheight=1.7in]{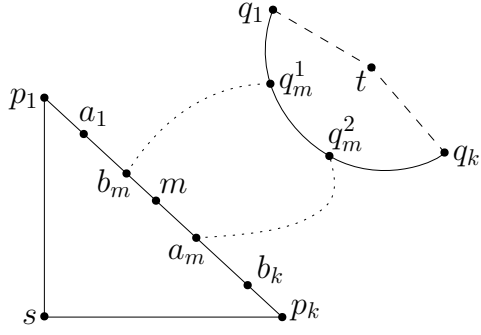}
\caption{\footnotesize
Illustrating a schematic view of the indices: $a_1$, $b_m$, $m$, $a_m$, and $b_k$.}
\label{fig:median}
\end{center}
\end{minipage}
\vspace*{-0.15in}
\end{figure}

Let $m=\lfloor(a_1+b_k)/2\rfloor$.
We compute a gateway in $V_t(1,k)$ that minimizes the value
$d(p_m,q)+d(q,t)$ for all $q\in V_t(1,k)$, and in case of a tie, we use
$q_m^1$ and $q_m^2$ to refer to the first and the last such gateways in
$V_t(1,k)$, respectively. Let $V_t(1,m)$ and $V_t(m,k)$ denote the
sublists of $V_t(1,k)$ from $q_1$ to $q_m^1$ and from $q_m^2$ to $q_k$,
respectively. We set one of $q^1_m$ and $q^2_m$ as the candidate
coupled gateway of $p_m$.

Define $a_m$ to be the largest index $i\in [m,b_k-1]$ such that
$d(p_{m},q^2_m)=d(p_{m},p_{i})+d(p_{i},q^2_m)$ and $b_m$ the smallest
index $i\in [a_1+1,m]$ such that $d(p_m,q^1_m)=d(p_m,p_{i})+d(p_{i},q^1_m)$. See Fig.~\ref{fig:median}.
We can compute $a_m$ and $b_m$ by a similar stair-walking
procedure as before.
According to Lemma~\ref{lem:enclose}, by similar proofs as
Lemma~\ref{lem:40}, we can show that for each $i\in [b_m,m-1]$, if $p_i$ is a via gateway, then $q^1_m$ is a coupled gateway of $p_i$, and
for each $i\in [m+1,a_m]$, if $p_i$ is a via gateway, then $q^2_m$ is a coupled gateway of $p_i$.
Thus we set $q^1_m$ as the candidate coupled gateway for each $p_i$ with $i\in [b_m,m-1]$, and set
$q^2_m$ as the candidate coupled gateway for each $p_i$ with $i\in [m+1,a_m]$.

If $a_m=b_k-1$ and $b_m=a_1+1$, then the candidate coupled gateways of all gateways in
$[1,k]$ have been computed and we can stop the algorithm.
If $a_m=b_k-1$ but $b_m>a_1+1$, the candidate coupled gateways of all gateways in
$[m,k]$ have been computed, and thus we work recursively on the subproblem $([a_1+1,b_m-1],[1,k],V_t(1,k))$ (note that the size of the first interval
is reduced by at least half). Similarly, if $b_m=a_1+1$ but $a_m<b_k-1$, then we work recursively on the subproblem $([a_m+1,b_k-1],[1,k],V_t(1,k))$. Otherwise, both $b_m>a_1+1$ and $a_m<b_k-1$ hold, and we proceed as follows.


We have the following two lemmas that are similar to Lemmas~\ref{lem:interior} and \ref{lem:70}.

\begin{lemma}\label{lem:110}
\begin{enumerate}
\item
The path $\pi_{q^2_m}(p_{a_m})$ contains a point in the interior of $R(s)$ only if
the last edge of the path intersects the bottom boundary of $R(s)$, in which the intersection at the bottom boundary of $R(s)$ has $x$-coordinate in $[x(s),x(p_{a_m+1})]$.
\item
The path $\pi_{q^1_m}(p_{b_m})$ contains a point in the interior of $R(s)$ only if the last edge of the path intersects the left boundary of $R(s)$, in which case the intersection at the left boundary of $R(s)$ has $y$-coordinate in $[y(s),y(p_{b_m-1})]$.
\end{enumerate}
\end{lemma}
\begin{proof}
The proof is similar to that for Lemma~\ref{lem:interior} and we omit the details. \qed
\end{proof}



\begin{lemma}\label{lem:120}
\begin{enumerate}
\item
If the last edge of $\pi_{q^2_m}(p_{a_m})$ intersects the bottom boundary of $R(s)$,
then $p_i$ cannot be a via gateway for any $i\in [a_m+1,b_k-1]$.
\item
If the last edge of $\pi_{q^1_m}(p_{b_m})$ intersects the left boundary of $R(s)$,
then $p_i$ cannot be a via gateway for any $i\in [a_1+1,b_m-1]$.
\end{enumerate}
\end{lemma}
\begin{proof}
The proof is similar to that of Lemma~\ref{lem:70}, but also relies on Lemma~\ref{lem:enclose}.
We briefly discuss it below.  We only prove the first part of the lemma
since the second part is similar. Let $e$ be the last edge of $\pi_{q^2_m}(p_{a_m})$. To simplify the notation, let $i=a_m$ and $q=q_m^2$.

\begin{figure}[t]
\begin{minipage}[t]{\linewidth}
\begin{center}
\includegraphics[totalheight=1.5in]{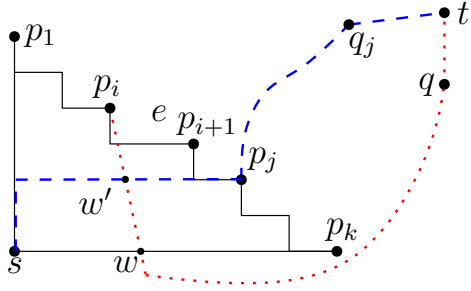}
\caption{\footnotesize Illustrating the the proof of Lemma~\ref{lem:120}: The (red) dotted path is $\pi(p_i,t)$ and the (blue) dashed path is $\pi(s,t)$. }
\label{fig:notpossible1}
\end{center}
\end{minipage}
\vspace*{-0.15in}
\end{figure}

Let $w$ be the
intersection of $e$ and the bottom boundary of $R(s)$. By Lemma~\ref{lem:110}, $x(w)\in [x(s),x(p_{i+1})]$. Assume to the contrary that $p_j$ for some $j\in [a_m+1,b_k-1]$ is a via
gateway. Then, by Lemma~\ref{lem:enclose}, there must be a shortest \st\
path $\pi(s,t)$ that contains $\overline{sp_j}$ and a gateway of $q_j$ in $V_t(1,k)$.
Without loss of generality, we assume that the sub-path of $\pi(s,t)$ between $s$ and $p_j$, denoted by $\pi(s,p_j)$, consists of a vertical segment through $s$ and a horizontal segment through $p_j$ (e.g., see Fig.~\ref{fig:notpossible1}).
Then, $\pi(s,p_j)$ intersects $e$ at a point, say, $w'$. Since $j\leq b_k-1<k$, $y(p_j)>y(p_k)=y(s)$. Thus, $y(w')>y(s)$.

Let $\pi(p_{i},t)$ denote the path $\pi_{q}(p_{i})\cup
\overline{qt}$, which contains $w'$.
Recall that $q$ is a gateway in $V_t(1,k)$ that minimizes the value
$d(p_{i},q')+d(q',t)$ for all $q'\in V_t(1,k)$. This implies that
$q$ is a gateway in $V_t(1,k)$ that minimizes the value
$d(w',q')+d(q',t)$ for all $q'\in V_t(1,k)$. Let $\pi'(w',t)$ be the
subpath of $\pi(p_{i},t)$ between $w'$ and $t$.

Let $\pi(w',t)$ be the sub-path of $\pi(s,t)$ between $w'$ and $t$. Since
$\pi(s,t)$ is a shortest \st\ path, $\pi(w',t)$ is also a shortest path
from $w'$ to $t$. Since $\pi(w',t)$ contains a gateway $q_j$ in
$V_t(1,k)$, $q_j$ is a gateway in $V_t(1,k)$ that minimizes the value
$d(w',q')+d(q',t)$ for all $q'\in V_t(1,k)$. Therefore,
the length of $\pi'(w',t)$ must be the same as that of
$\pi(w',t)$.
Hence, if we replace the subpath $\pi(w',t)$ of $\pi(s,t)$ by
$\pi'(w',t)$, we obtain another shortest \st\ path $\pi'(s,t)$.

Notice that the sub-path of $\pi'(s,t)$ between $s$ and $w$ is the concatenation of the sub-path of $\pi(s,p_j)$ from $s$ to $w'$ and $\overline{w'w}$, whose length is
strictly larger than $|\overline{sw}|$ because $y(w')>y(s)=y(w)$.
However, since $d(s,w)=|\overline{sw}|$,
$\pi'(s,t)$ cannot be a shortest path. Thus we obtain contradiction.
\qed
\end{proof}

In constant time we can check whether the two cases in
Lemma~\ref{lem:120} happen. If both cases happen, then we can stop the
algorithm. If the second case happens and the first one does not, then
we recursively work on the subproblem $([a_m+1,b_k-1],[1,k],V_t(1,k))$.
If the first case happens and the
second one does not, then we recursively work on the subproblem $([a_1+1,b_m-1],[1,k],V_t(1,k))$.
In the following, we assume that neither case happens.
By Lemma~\ref{lem:110}, neither $\pi_{q^2_m}(p_{a_m})$ nor
$\pi_{q^1_m}(p_{b_m})$ contains a point in the interior of $R(s)$.
Consequently, we have the following lemma.

\begin{lemma}\label{lem:enclose1}
\begin{enumerate}
\item
For each $i\in [{a_m+1},b_k-1]$, if $p_i$ is a via gateway, then $p_i$ has a coupled gateway in $V_t(m,k)$.
If $q_m^2\neq q_k$, then $\pi_{q_m^2}(p_{a_m})$ does not intersect $\pi_{q_k}(p_{b_k})$.
\item
For each $i\in [{a_1+1},b_m-1]$, if $p_i$ is a via gateway, then $p_i$ has a coupled gateway in $V_t(1,m)$.
If $q_m^1\neq q_1$, then $\pi_{q_m^1}(p_{b_m})$ does not intersect $\pi_{q_1}(p_{a_1})$.
\end{enumerate}
\end{lemma}
\begin{proof}
We only prove the first part of the lemma, since the second part is
similar. Suppose $p_i$ is via gateway with $i\in [{a_m+1},b_k-1]$.
Then, there is a shortest \st\ path $\pi(s,t)$ that contains $p_i$, and let
$\pi(p_i,t)$ be the subpath between $p_i$ and $t$. By
Lemma~\ref{lem:optint}, $\pi(p_i,t)$ does not contain any interior
point of $R(s)$.

We first assume that $q_m^2\neq q_1$. Due to Observation~\ref{obser:tie}, we claim that the path
$\pi_{q_m^2}(p_{a_m})$ does not intersect the path
$\pi_{q_1}(p_{a_1})$. Indeed, assume to the contrary that the two
paths intersect, say, at the point $w$. Then, by the definitions of
$q_1$ and $q_2^m$, each of them is a point in $V_t(1,k)$ minimizing
the value $d(w,q)+d(q,t)$ for all $q\in V_t(1,k)$. This means that
$d(p_{a_1},q_1)+d(q_1,t)=d(p_{a_1},q_m^2)+d(q_m^2,t)$. However, this
contradicts with Observation~\ref{obser:tie} since $q_m^2\neq q_1$ and
$q_m^2\in V_t(1,k)$.

\begin{figure}[t]
\begin{minipage}[t]{\linewidth}
\begin{center}
\includegraphics[totalheight=1.7in]{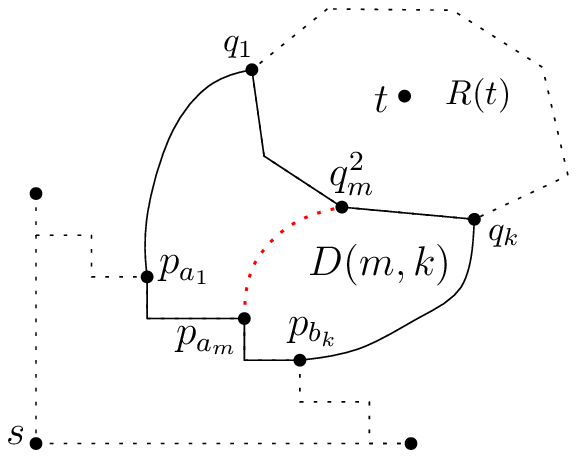}
\caption{\footnotesize
The region bounded by the sold curves is $D(1,k)$, and $\pi_{q_m^2}(p_{a_m})$ (the red dotted curve) partitions it into two subregions, one of which is $D(m,k)$.}
\label{fig:cycle1}
\end{center}
\end{minipage}
\vspace*{-0.15in}
\end{figure}

Depending on whether $q_m^2$ is $q_k$, there are two cases.

If $q_m^2\neq q_k$, then by the similar proof as above, the path
$\pi_{q_m^2}(p_{a_m})$ does not intersect $\pi_{q_k}(p_{b_k})$ either.
Recall that we have defined a region $D(1,k)$ that is bounded by $\beta_{s}[p_{a_1},p_{b_k}]$, $\pi_{q_1}(p_{a_1})$, $\pi_{q_k}(p_{b_k})$, and a boundary portion $B_t(1,k)$ of $R(t)$, e.g., see Fig.~\ref{fig:cycle1}.
Recall that $\pi_{q_m^2}(p_{a_m})$ does not intersect the interior of
$R(s)$. Since $p_{a_m}\in \beta_{s}[p_{a_1},p_{b_k}]$, $\pi_{q_m^2}(p_{a_m})$ does not intersect either
$\pi_{q_1}(p_{a_1})$ or $\pi_{q_k}(p_{b_k})$, and $t$ is not in $D(1,k)$, if $w$ is the first
point of $\pi_{q_m^2}(p_{a_m})$ on $\partial R(t)$ (such a point $w$
must exists since $q_m^2$ is on $\partial R(t)$), then $w$ must be on
$B_t(1,k)$. By our way of defining $B_t(1,k)$ and according to Lemma~\ref{lem:extend}(5),
the sub-path of $\pi_{q_m^2}(p_{a_m})$ between $w$ and $q_m^2$ is $\overline{wq_m^2}$, which must be on $B_t(1,k)$.
This implies that $\pi_{q_m^2}(p_{a_m})$ is in $D(1,k)$.
Since both endpoints of $\pi_{q_m^2}(p_{a_m})$ are on the boundary of $D(1,k)$, $\pi_{q_m^2}(p_{a_m})$ partitions $D(1,k)$ into two subregions, one of which, denoted by $D(m,k)$, contains $\beta_s[p_{a_m},p_{b_k}]$. Let $B_t(m,k)$ denote the portion of $B_t(1,k)$ in $D(m,k)$. By definition, $V_t(m,k)= V(t)\cap B_t(m,k)$. Recall that both $q_m^2$ and $q_k$ are in $V_t(m,k)$.

We proceed to show that $V_t(m,k)$ contains a coupled gateway of $p_i$.
If the path $\pi(p_i,t)$ intersects $\pi_{q_m^2}(p_{a_m})$, then by the similar analysis as before, $q_m^2$ is a coupled gateway of $p_i$. Similarly, if $\pi(p_i,t)$ intersects $\pi_{q_k}(p_{b_k})$, then $q_k$ is a coupled gateway of $p_i$. In the following, we assume that $\pi(p_i,t)$ does not intersect either path.
Recall that the path $\pi(p_i,t)$ does not contain any interior point of $R(s)$.
Since $p_i$ is in $\beta_s[p_{a_m},p_{b_k}]$ (and thus is in $D(m,k)$)
but $t$ is not in $D(m,k)$, $\pi(p_i,t)$ must intersect $B_t(k,m)$,
say, at a point $w$. By our way of defining $B_t(1,k)$ and according
to Lemma~\ref{lem:extend}, $B_t(k,m)$ contains a gateway $q$ such that
$\overline{wq}\cup \overline{qt}$ is a shortest path from $w$ to $t$.
This implies that $q$ is a coupled gateway of $p_i$.
Since $q\in B_t(m,k)$ and $V_t(m,k)= V(t)\cap B_t(m,k)$, $q$ is in $V_t(m,k)$. The lemma is thus proved.

Next, we consider the case where $q_m^2=q_k$. In this case, $V_t(m,k)=\{q_k\}$ and our goal is to show that $q_k$ is a coupled gateway of $p_i$. If we move on $\pi_{q_m^2}(p_{a_m})$ from $p_{a_m}$ to $q_m^2$, let $w$ be the first intersection of $\pi_{q_m^2}(p_{a_m})$ and $\pi_{q_k}(p_{b_k})$. Let $\pi(p_{a_m},w)$ be the sub-path of $\pi_{q_m^2}(p_{a_m})$ between $p_{a_m}$ and $w$, and $\pi(p_{b_k},w)$ the sub-path of $\pi_{q_k}(p_{b_k})$ between $p_{b_k}$ and $w$.
Recall that $\pi_{q_m^2}(p_{a_m})$ does not intersect
$\pi_{q_1}(p_{a_1})$ and does not contain any interior point of
$R(s)$. We claim that $\pi(p_{a_m},w)$ is contained in the region
$D(1,k)$. Indeed, this is obviously true if $\pi(p_{a_m},w)$ does not
intersect $\partial R(t)$. Otherwise, let $z$ be the first
intersection between $\pi(p_{a_m},w)$ and $\partial R(t)$. Note that
$z$ must be on $B_t(1,k)$. According to Lemma~\ref{lem:extend}(5), the
sub-path of $\pi_{q_m^2}(p_{a_m})$ between $z$ and $q_m^2$ must be the
segment $\overline{zq_m^2}$, which is on $B_t(1,k)$. This also implies
that $w\in \overline{zq_m^2}$ and $\pi(p_{a_m},w)$ is in $D(1,k)$, and
further, $\pi(p_{a_m},w)$ does not contain any point in the interior
of $R(t)$. Let $D$ be the sub-region of $D(1,k)$ bounded by
$\pi(p_{a_m},w)$, $\pi(p_{b_k},w)$, and $\beta_{s}[p_{a_m},p_{b_k}]$.
Clearly, $D$ does not contain $t$.

Now consider the path $\pi(p_i,t)$. Since $p_i\in \beta_{s}[p_{a_m},p_{b_k}]\subseteq D$, $t\not\in D$,
$\pi(p_i,t)$ does not contain any interior point of $R(s)$, and
neither $\pi(p_{a_m},w)$ nor $\pi(p_{b_k},w)$ contains an interior
point of $R(t)$,  $\pi(p_i,t)$ must intersect either $\pi(p_{a_m},w)$ or
$\pi(p_{b_k},w)$ (and thus intersect either $\pi_{q_m^2}(p_{a_m})$ or $\pi_{q_k}(p_{b_k})$). In either case, by the similar analysis as above, $q_k$ ($=q_m^2$) is a coupled gateway of $p_i$. The lemma is thus proved.

The above prove the case where $q_m^2\neq q_1$. If $q_m^2=q_1$, then $V_t(m,k)=V_t(1,k)$. By Lemma~\ref{lem:enclose}, it is trivially true that $p_i$ has a coupled gateway in $V_t(m,k)$.
Further, due to Observation~\ref{obser:tie}, by similar analysis as before, $\pi_{q_m^2}(p_{a_m})$ cannot intersect $\pi_{q_k}(p_{b_k})$. The lemma thus follows.
\qed
\end{proof}

Based on Lemma~\ref{lem:enclose1}, our algorithm proceeds as follows.
If $q_m^2=q_k$, then we set $q_k$ as the candidate coupled gateway for each $p_i$ with $i\in [a_m+1,b_k-1]$. Otherwise, we call the algorithm recursively on the subproblem $([a_m+1,b_k-1],[m,k],V_t(m,k))$. Similarly, if $q_m^1=q_1$, then we set $q_1$ as the candidate coupled gateway for each $p_i$ with $i\in [a_1+1,b_m-1]$. Otherwise, we call the algorithm recursively on the subproblem $([a_1+1,b_m-1],[1,m],V_t(1,m))$.

For the running time, notice that the stair-walking procedure
spends $O(1)$ time on finding a coupled gateway for a gateway of
$V(s)$. Hence, the overall time of the stair-walking procedure in the entire
algorithm is $O(n_s)$. Consider a subproblem $([i',j'], [i,j], V_t(i,j))$.To solve it, after spending $O(|V_t(i,j)|)$ time, we
either reduce the problem to another subproblem in which the first
interval is at most half the size of $[i',j']$ and the third gateway set is
still $V_t(i,j)$, or reduce it to two sub-problems such that
each of them has the first interval at most half the size of $[i',j']$ and the third gateway
sets of the two sub-problems are two disjoint subsets of $V_t(i,j)$. Hence, if we consider the algorithm procedure as a tree structure, the height of the tree is $O(\log n_s)$ and the total time we spend on each level of the tree is $O(n_t)$. Therefore, the overall time of the algorithm is $O(n_s+n_t\log n_s)$.

\subsubsection{The equal case $q_1=q_k$}
\label{sec:equal}

For the case $q_1=q_k$, we will eventually reduce it to the above unequal case.
In this case, we will need to determine the relative positions of two
shortest paths (e.g., $\pi_{q_1}(p_{a_1})$ and $\pi_{q_1}(p_{b_k})$)
with respect to $\overline{q_1t}$. To this end, we perform the
following additional preprocessing.

Recall that we have already computed a shortest path tree $T(q_1)$ from
$q_1$ to all vertices of $G$.
In addition, we compute a post-order traversal list on $T(q_1)$ (but
excludes the root $q_1$) and store the list in a cyclic array $L(q_1)$.
This does not change the preprocessing complexities asymptotically.

Recall that $t$ is visible to $q_1$. We want to know the position of
$t$ at $L(q_1)$ if we ``insert'' $t$ into the tree $T(q_1)$ (and thus
$t$ becomes a leaf). This can
be done in $O(\log n)$ time by doing binary search on the children
of $q_1$ in $T(q_1)$. After that, given any two vertices $v_1$ and
$v_2$ of $T(q_1)$, by using $L(q_1)$, we can determine in constant time whether
$\pi_{q_1}(v_1)$ is clockwise from $\pi_{q_1}(v_2)$ with
respect to the path $\pi_{q_1}(t)=\overline{q_1t}$ (similar approach was also used in~\cite{ref:MitchellSh96}; for simplicity, we
assume that $v_2\not \in \pi_{q_1}(v_1)$ and $v_1\not\in
\pi_{q_1}(v_2)$, which is also the case in our algorithm; we say that
$\pi_{q_1}(v_2)$ is {\em clockwise from} $\pi_{q_1}(v_1)$ if we meet
$\pi_{q_1}(v_1)$ first when topologically rotating $\overline{q_1t}$
around $q_1$
clockwise; e.g., see Fig.~\ref{fig:relpos}).

\begin{figure}[t]
\begin{minipage}[t]{0.43\linewidth}
\begin{center}
\includegraphics[totalheight=1.1in]{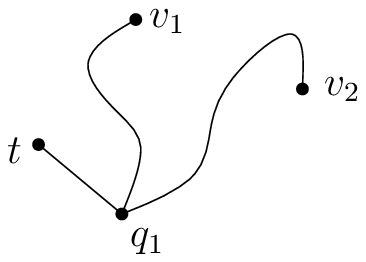}
\caption{\footnotesize Illustrating an example in which $\pi_{q_1}(v_2)$ is clockwise from $\pi_{q_1}(v_1)$ with respect to $\overline{q_1t}$.}
\label{fig:relpos}
\end{center}
\end{minipage}
\hspace*{0.05in}
\begin{minipage}[t]{0.56\linewidth}
\begin{center}
\includegraphics[totalheight=1.7in]{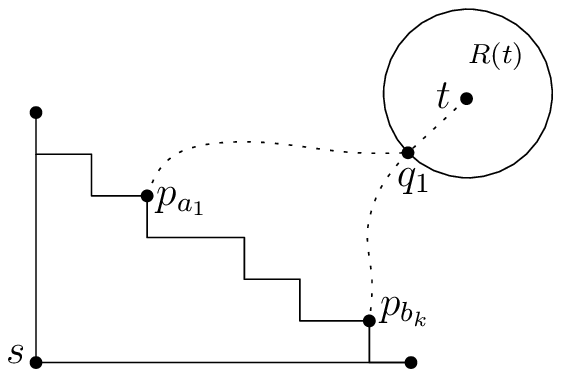}
\caption{\footnotesize Illustrating the case where $\pi_{q_1}(p_{a_m})$ is clockwise from $\pi_{q_1}(p_{b_k})$ with respect to $\overline{q_1t}$.}
\label{fig:clockwise1}
\end{center}
\end{minipage}
\vspace*{-0.15in}
\end{figure}

We first check whether $\pi_{q_1}(p_{a_1})$ is clockwise from
$\pi_{q_1}(p_{b_k})$ with respect to $\overline{q_1t}$.  If yes,
the following lemma implies that we can stop our algorithm by
setting $q_1$ as a candidate coupled gateway for all $p_i$ with $i\in
[a_1+1,b_k-1]$.

\begin{lemma}\label{lem:clockwise}
If $\pi_{q_1}(p_{a_1})$ is clockwise from $\pi_{q_1}(p_{b_k})$ with respect to $\overline{q_1t}$ (e.g., see Fig.~\ref{fig:clockwise1}), then for each $i\in [a_1+1,b_k-1]$, if $p_i$ is a via gateway, then $q_1$ is a coupled gateway of $p_i$.
\end{lemma}
\begin{proof}
If we move from $p_{a_1}$ to $q_1$ on $\pi_{q_1}(p_{a_1})$, let $w$ be the first point of the path that intersects $\pi_{q_1}(p_{b_k})$. Let $\pi(p_{a_1},w)$ denote the subpath of $\pi_{q_1}(p_{a_1})$ between $p_{a_1}$ and $w$, and $\pi(p_{b_k},w)$ the subpath of $\pi_{q_1}(p_{b_k})$ between $p_{b_k}$ and $w$.
Since neither $\pi_{q_1}(p_{a_1})$ nor $\pi_{q_1}(p_{b_k})$ contains any interior point of $R(s)$, $\pi(p_{a_1},w)\cup \pi(p_{b_k},w)\cup \beta_s[p_{a_1},p_{b_k}]$ forms a closed cycle that divides the plane into two regions. We use $D$ to denote the region that does not contain $s$. Since $\pi_{q_1}(p_{a_1})$ is clockwise from $\pi_{q_1}(p_{b_k})$ with respect to $\overline{q_1t}$ and $p_{a_1}$ is counterclockwise from $p_{b_k}$ on $\beta_s[p_{a_1},p_{b_k}]$ with respect to $s$, the region $D$ does not contain $t$.
Further, by Lemma~\ref{lem:disconnect}, $D$ does not contain any interior point of $R(t)$ and contains at most one (i.e., $q_1$ if $w=q_1$) gateway of $t$.

Suppose $p_i$ is a via gateway with $i\in [a_1+1,b_k-1]$. There is a shortest \st\ path containing $p_i$, and we use $\pi(p_i,t)$ to denote the subpath between $p_i$ and $t$.
By Lemma~\ref{lem:optint}, $\pi(p_i,t)$ does not contain any interior point of $R(s)$. Since $i\in [a_1+1,b_k-1]$, $p_i\in \beta_s[p_{a_1},p_{b_k}]$. As $t\not\in D$, $\pi(p_i,t)$ must intersect either $\pi(p_{a_1},w)$ or $\pi(p_{b_k},w)$. In either case, by similar analysis as before (e.g., in Lemma~\ref{lem:enclose}), we can show that $q_1$ is a coupled gateway of $p_i$, and we omit the details.
\qed
\end{proof}

If $\pi_{q_1}(p_{a_1})$ is counterclockwise from $\pi_{q_1}(p_{b_k})$, then we proceed as follows.

Let $m=\lfloor(a_1+b_k)/2\rfloor$.
We compute a gateway in $V(t)$ that minimizes the value $d(p_m,q)+d(q,t)$ for all $q\in V(t)$, and in case of tie, we use $q_m^1$ to refer to the first one in $V(t)$ in the counterclockwise order from $q_1$,
and use $q_m^2$ to refer to the first one in $V(t)$ in the clockwise order from $q_1$.
We set one of $q^1_m$ and $q^2_m$ as the candidate coupled gateway of $p_m$.
Note that $q_m^1\neq q_1$ if and only if $q_m^2\neq q_1$.
Depending on whether $q_m^1=q_1$, there are two cases.

If $q_m^1\neq q_1$ (and thus $q_m^2\neq q_2$), then we apply our algorithm for the above unequal case on $[1,m]$ and the gateways of $V(t)$ from $q_1$ to $q_m^1$ in the counterclockwise order. We also apply the algorithm on $[m,k]$ and the gateways of $V(t)$ from $q_k$ to $q_m^2$ in the clockwise order. Therefore, in this case, we have reduced our problem to the unequal case.

If $q_m^1=q_1$, then $q_m^2=q_1$. In this case, we work on the problem
for the equal case recursively until the subproblems are reduced to
the unequal case (and then we apply the unequal case algorithm).
Each recursive step works on a subproblem, denoted by
$([i',j'], [i,j], V(t))$ with $[i',j']\subseteq [i,j]\subseteq [1,k]$,
where we want to find the candidate coupled gateways in the interval
$[i',j']$, $q_1$ is a coupled gateway for both $p_{a_i}$ and
$p_{b_j}$, and $\pi_{q_1}(p_{a_i})$ is counterclockwise from
$\pi_{q_1}(p_{b_j})$. Initially, our subproblem is
$([a_1+1,b_k-1],[1,k], V(t))$. We proceed as follows.

Define $a_m$ and $b_m$ in the same way as before in the unequal case. Similarly as before, if $a_m=b_k-1$ but $b_m>a_1+1$, the candidate coupled gateways of $p_i$ for all
$i\in [m,k]$ have been computed, and thus we work recursively on the
subproblem $([a_1+1,b_m-1],[1,k],V(t))$; if $b_m=a_1+1$ but
$a_m<b_k-1$, then we work recursively on the subproblem
$([a_m+1,b_k-1],[1,k],V(t))$. Otherwise both $b_m>a_1+1$ and
$a_m<b_k-1$ hold, and we proceed as follows.

Note that Lemmas~\ref{lem:110} and \ref{lem:120} still hold.
In constant time we can check whether the two cases in
Lemma~\ref{lem:120} happen. If both cases happen, then we can stop the
algorithm. If the second case happens but the first one does not, then
we recursively work on the subproblem $([a_m+1,b_k-1],[1,k],V(t))$. If the first case happens but the
second one does not, then we recursively work on the subproblem $([a_1+1,b_m-1],[1,k],V(t))$.
In the following, we assume that neither case happens.
By Lemma~\ref{lem:110}, neither $\pi_{q_1}(p_{a_m})$ nor
$\pi_{q_1}(p_{b_m})$ contains a point in the interior of $R(s)$.

\begin{figure}[t]
\begin{minipage}[t]{\linewidth}
\begin{center}
\includegraphics[totalheight=2.0in]{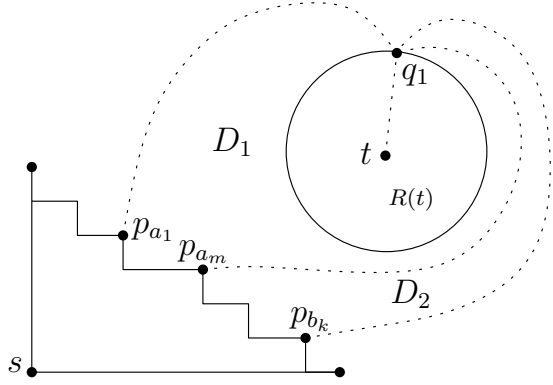}
\caption{\footnotesize Illustrating the case where $\pi_{q_1}(p_{a_m})$ is clockwise from $\pi_{q_1}(p_{b_k})$ with respect to $\overline{q_1t}$.}
\label{fig:clockwise}
\end{center}
\end{minipage}
\vspace*{-0.15in}
\end{figure}

In constant time, we further check whether $\pi_{q_1}(p_{a_m})$ is clockwise from $\pi_{q_1}(p_{b_k})$ with respect to $\overline{q_1t}$.
We have the following lemma.

\begin{lemma}\label{lem:enclose2}
Let $\pi$ be either $\pi_{q_1}(p_{a_m})$ or $\pi_{q_1}(p_{b_m})$.
If $\pi$ is clockwise from $\pi_{q_1}(p_{b_k})$ with respect to
$\overline{q_1t}$ (e.g., see Fig.~\ref{fig:clockwise}), then for each
$i\in [{a_m+1},b_k-1]$, if $p_i$ is a via gateway, then $q_1$ is a
coupled gateway of $p_i$.
Otherwise, for each $i\in [{a_1+1},b_m-1]$, if $p_i$ is a via gateway,
then $q_1$ is a coupled gateway of $p_i$.
\end{lemma}
\begin{proof}
We only prove the case where $\pi$ is $\pi_{q_1}(p_{a_m})$, since the other case is similar.

Note that $\pi_{q_1}(p_{a_1})$ and $\pi_{q_1}(p_{b_k})$ do not cross each other because they are paths in the shortest tree $T(q_1)$.
Since neither $\pi_{q_1}(p_{a_1})$ nor $\pi_{q_1}(p_{b_k})$ contains
any interior point of $R(s)$, the two paths along with
$\beta[p_{a_1},p_{b_k}]$ form a closed cycle that divides the plane into two regions, one
of which (denoted by $D$) does not contain $s$. Since $\pi_{q_1}(p_{a_1})$ is counterclockwise from
$\pi_{q_1}(p_{b_k})$ with respect to $\overline{q_1t}$, $p_{a_1}$ is counterclockwise from $p_{b_k}$ on $\beta_s[p_{a_1},p_{b_k}]$ with respect to $s$, and neither $\pi_{q_1}(p_{a_1})$ nor $\pi_{q_1}(p_{b_k})$ contains any interior point of $R(t)$ (by Lemma~\ref{lem:disconnect}), $D$ contains
$R(t)$.

Recall that $\pi_{q_1}(p_{a_m})$ does not contain any interior point of
$R(s)$. Also, $\pi_{q_1}(p_{a_m})$ does not cross either
$\pi_{q_1}(p_{a_1})$ or $\pi_{q_1}(p_{b_k})$ since they are paths in the shortest path tree $T(q_1)$.
Since both endpoints of
$\pi_{q_1}(p_{a_m})$  are on the boundary of $D$,
$\pi_{q_1}(p_{a_m})$ partitions $D$ into two subregions  (e.g., see Fig.~\ref{fig:clockwise}): One
subregion, denoted by $D_1$, is bounded by $\pi_{q_1}(p_{a_m})$,
$\pi_{q_1}(p_{a_1})$, and $\beta_s[p_{a_1},p_{a_m}]$, and the other,
denoted by $D_2$, is bounded by $\pi_{q_1}(p_{a_m})$,
$\pi_{q_1}(p_{b_k})$, and $\beta_s[p_{a_m},p_{b_k}]$.
In addition, by the similar analysis, we can show that Lemma~\ref{lem:disconnect} also applies to the path $\pi_{q_1}(p_{a_m})$.

If $\pi_{q_1}(p_{a_m})$ is clockwise from $\pi_{q_1}(p_{b_k})$ with
respect to $\overline{q_1t}$, then $R(t)$ must be contained in $D_1$
(e.g., see Fig.~\ref{fig:clockwise}). Suppose $p_i$ is a via gateway
with $i\in [a_m+1,b_k-1]$. Then, there is a shortest \st\ path containing $p_i$, and we use $\pi(p_i,t)$ to denote the subpath between $p_i$ and $t$.
By the same analysis as that in Lemma~\ref{lem:clockwise}, we can show that $\pi(p_i,t)$ must  intersect either $\pi_{q_1}(p_{a_m})$ or $\pi_{q_1}(p_{b_k})$. In either case, $q_1$ is a coupled gateway of $p_i$.

If $\pi_{q_1}(p_{a_m})$ is counterclockwise from $\pi_{q_1}(p_{b_k})$
with respect to $\overline{q_1t}$, then $R(t)$ must be contained in
$D_2$. Then, by similar analysis as above, we can show that
for each $i\in [{a_1+1},b_m-1]\subseteq [{a_1+1},a_m-1]$, if $p_i$ is a via gateway,
then $q_1$ is a coupled gateway of $p_i$.
We omit the details.
\qed
\end{proof}

By Lemma~\ref{lem:enclose2}, depending on whether $\pi_{q_1}(p_{a_m})$ is clockwise from $\pi_{q_1}(p_{b_k})$, there are two cases.

\begin{enumerate}
\item
If yes, then we set $q_1$ as the candidate coupled gateway for all $p_i$ with $i\in [{a_m+1},b_k-1]$. Depending on whether $\pi_{q_1}(p_{b_m})$ is counterclockwise from $\pi_{q_1}(p_{b_k})$, there are further two subcases.

\begin{enumerate}
\item
If yes, we set $q_1$ as the candidate coupled gateway for all $p_i$ with
$i\in [{a_1+1},b_m-1]$. Note that we have found the candidate coupled gateways for all $p_i$ with $i\in [a_1+1,b_k-1]$. Hence, we can stop the algorithm.

\item
Otherwise, we recursively work on the subproblem $([a_1+1,b_m-1],[1,m],V(t))$.
\end{enumerate}

\item

If $\pi_{q_1}(p_{a_m})$ is counterclockwise from $\pi_{q_1}(p_{b_k})$,
then we set $q_1$ as the candidate coupled gateway for all $p_i$ with
$i\in [{a_1+1},b_m-1]$. Depending on whether $\pi_{q_1}(p_{a_m})$ is
clockwise from $\pi_{q_1}(p_{b_k})$, there are further two
subcases.

\begin{enumerate}
\item
If yes, we set $q_1$ as a candidate coupled gateway for all $p_i$ with $i\in [{a_m+1},b_k-1]$. Then, we stop the algorithm.

\item
Otherwise, we recursively work on the subproblem $([a_m+1,b_k-1],[m,k],V(t))$.
\end{enumerate}
\end{enumerate}

In this way, we have either computed candidate gateways for all
gateways of $V(s)$ or reduced the problem to the unequal case. Note
that each recursive step reduces the length of the first
interval of the subproblem by half in $O(n_t)$ time. In addition, the
total time for the stair-walking procedure is $O(n_s)$. Therefore, the
total time of the algorithm for handling the equal case is $O(n_s+n_t\log n_s)$.

\subsection{Wrapping Up}

The above describes our algorithm on the gateways of $s$ in the first quadrant of $s$. We run the same algorithm for all quadrants of $s$, and for each quadrant, we will find an \st\ path. Finally, we return the path with the smallest length as our solution. The proof of the following lemma summarizes our entire query algorithm.

\begin{lemma}\label{lem:query22}
The running time of the query algorithm is $O(\log n+ n_s+n_t\log n_s)$.
\end{lemma}
\begin{proof}
Given $s$ and $t$, we first check whether there is a trivial shortest path. If not, we compute the gateway sets $V_g(s,G)$ and $V_g(t,G)$. We then explicitly compute the gateway region $R(t)$. Let $V(t)$ be the gateways on the boundary of $R(t)$, as defined before, including those special gateways. All above can be computed in $O(\log n)$ time.

Next, we compute the gateway $p_1\in V_g^1(s,G)$ that minimizes the value $\min_{q\in V(t)}(d(s,p)+d(p,q)+d(q,t))$ among all $p\in V_g^1(s,G)$, which can be done in $O(n_t)$ time since $|V(t)|=O(n_t)$.

Then, we apply our algorithm in this section on $V_g^2(s,G)$ and $V(t)$, which will return a gateway $p_2\in V_g^2(s,G)$ such that if $V_g^2(s,G)$ contains a via gateway, then $p_2$ is a via gateway. This takes $O(\log n+ n_s+n_t\log n_s)$ time.

For each $i=1,2$, let $d_i=\min_{q\in V(t)}(d(s,p_i)+d(p_i,q)+d(q,t))$ and let $q_i$ be the gateway of $V(t)$ such that $d_i=d(s,p_i)+d(p_i,q_i)+d(q_i,t)$. Without loss of generality, we assume $d_1\leq d_2$. Then, $d(s,t)=d_1$. Using the shortest path tree $T(q_1)$, we can find a shortest path from $p_1$ to $q_1$ in linear time in the number of edges of the path, and then by appending $\overline{sp_1}$ and $\overline{q_1t}$ we can obtain a shortest \st\ path.
\qed
\end{proof}

Since both $n_s$ and $n_t$ are $O(\log n)$, we have the following corollary.

\begin{corollary}
With $O(n^2\log^3 n)$ time and $O(n^2\log^2 n)$ space preprocessing,
given any two query points $s$ and $t$, we can
compute their shortest path length in $O(\log n\log\log n)$ time and an
actual shortest \st\ path can be output in additional time linear in the number of
edges of the path.
\end{corollary}

\section{Reducing the Query Time to $O(\log n)$}
\label{sec:overall}

To further reduce the query time to $O(\log n)$, we need to change our
graph $G$ to a slightly larger graph $G_1$ such that $t$ only needs $O(\log
n/\log\log n)$ gateways while $s$ still has $O(\log n)$ gateways,
i.e., $n_s=O(\log n)$ and $n_t=O(\log n/\log\log n)$. To this end, we
introduce more Steiner points on the cut-lines. A
similar idea was also used in~\cite{ref:ChenTw16} to reduce the number
of gateways to $O(\sqrt{\log n})$. However, since we are allowed to
have more gateways than $O(\sqrt{\log n})$, we do not need as many Steiner
points  as those in~\cite{ref:ChenTw16}, which is the reason
why we use less preprocessing.

Specifically, comparing with $G$, the new graph $G_1$ has the
following changes. As in~\cite{ref:ChenTw16}, we first define
``super-levels''. Recall that the cut-line tree $\calT$ has $O(\log
n)$ levels (with the root at the first level). We further partition
all levels of the tree into $O(\log n/\log\log n)$ super-levels: For
any $i$, the $i$-th super-level contains the levels from $(i-1)\cdot
\log\log n+1$ to $i\cdot \log\log n$. Hence, each super-level has
at most $\log\log n$ levels.

\begin{figure}[t]
\begin{minipage}[t]{\linewidth}
\begin{center}
\includegraphics[totalheight=1.6in]{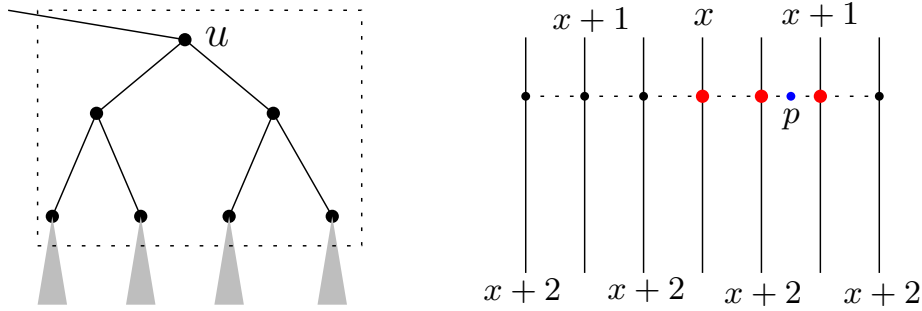}
\caption{\footnotesize Left: Illustrating the subtree $\calT_u$, which is in the dotted rectangle (we assume  $\log\log n = 3$). Right: Illustrating the type-3 Steiner points defined by a point $p$ on $\calT_u$. The vertical lines are the cut-lines of the nodes in $\calT_u$ and their level numbers are also shown (we assume that the level number of $u$ is $x$). We assume that $p\in \calV(u)$ and $p$ is horizontally visible to all these cut-lines. Then, $p$ defines a type-3 Steiner point on each cut-line.
In contrast, only the three big (red) points are type-2 Steiner points defined by $p$ in our original graph $G$}
\label{fig:superlevel}
\end{center}
\end{minipage}
\vspace{-0.15in}
\end{figure}

Let $u$ be a node at the highest level of the $i$-th super level of
$\calT$. Let $\calT_u$ be the sub-tree of $\calT$ rooted at $u$
excluding the nodes outside the $i$-th level (thus $\calT_u$ has at most $\log n-1$ nodes); e.g., see Fig~\ref{fig:superlevel}. Recall that $u$ is
associated with a subset $\calV(u)$ of polygon vertices and each
vertex $v\in \calT_u$ is associated with a cut-line $l(v)$.
For each point $p\in \calV(u)$ and each vertex $v\in \calT_u$, if $p$
is horizontally visible to $l(v)$, then $p$ defines a {\em type-3}
Steiner point on $l(v)$. In this way, $p$ defines $O(\log n)$ type-3
Steiner points on the cut-lines of $\calT_u$ (in contrast, $p$
defines only $O(\log\log n)$ type-2 Steiner points on the cut-lines of
$\calT_u$ in our original graph $G$); e.g., see Fig~\ref{fig:superlevel}. Hence, each polygon vertex
$p$ defines a total of $O(\log^2 n/\log\log n)$ type-3 Steiner points
since $\calT$ has $O(\log n/\log\log n)$ super-levels. The
total number of type-3 Steiner points on all cut-lines is $O(n\log^2 n/\log\log n)$.
Note that each type-2 Steiner point in our original graph $G$ becomes
a type-3 Steiner point. For convenience of discussion, those type-3
Steiner points of $G_1$ that are originally type-2 Steiner points of
$G$ are also called type-2 Steiner points of $G_1$.

Type-1 Steiner points are defined in the same way as before, so their number is still $O(n)$.
We still use $\calV_1$ to denote the set of all type-1 Steiner points
and all polygon vertices.  We use $\calV_2$ to denote the set of all
type-2 Steiner points of $G_1$.

The edges of $G_1$ are defined with respect to all Steiner points in
the same way as $G$. We omit the details. In summary,  $G_1$ has
$O(n\log^2 n/\log\log n)$ vertices and edges. $G_1$ can be
computed in $O(n\log^3 n/\log\log n)$ time (e.g., by using the similar
algorithm as in the proof of Lemma~1 in~\cite{ref:ChenTw16}). Note that
the original graph $G$ is a sub-graph of $G_1$ in that
every vertex of $G$ is also a vertex of $G_1$ and every path of $G$
corresponds to a path in $G_1$ with the same length.

Consider a query point $t$. The gateway set $V^1_g(t,G_1)$ is defined in the same way as before, and thus its size is $O(1)$. Thanks to more Steiner points, the size of $V^2_g(t,G_1)$ can now be reduced to $O(\log n/\log\log n)$. Specifically,  $V^2_g(t,G_1)$ is defined as follows (similar to that in~\cite{ref:ChenTw16}).

As in~\cite{ref:ChenTw16}, we first define the {\em relevant projection
cut-lines} of $t$. We only discuss the right side of $t$,
and the left side is symmetric.
Recall that $t$ has at most one
projection cut-line in each level of $\calT$. Among all projection
cut-lines that are in the same super-level, the one closest to $t$ is
called a {\em relevant projection cut-line} of $t$. Since there are
$O(\log n/\log\log n)$ super-levels and each super-level has at most
one relevant projection cut-line to the right of $t$, $t$ has
$O(\log n/\log\log n)$ relevant projection cut-lines. For each such
cut-line $l$, the Steiner point (if any) immediately
above (resp., below) the horizontal projection $t'$ of $t$ on $l$ is
included in $V_g^2(t,G_1)$ if it is visible to $t'$. Thus, $|V_g^2(t,G_1)|=O(\log n/\log\log
n)$.

By Lemma~\ref{lem:query}, if $|V_g^2(t,G_1)|=O(\log n/\log\log n)$,
the query time becomes $O(\log n)$ as long as $|V_g^2(s,G_1)|=O(\log
n)$. This implies that for $s$, we can simply use its original gateway
set on type-2 Steiner points, i.e., we define $V_g^2(s,G_1)$ in the same
way as before with respect to only the type-2 Steiner points of $G_1$ (thus $V_g^2(s,G_1)=V_g^2(s,G)$). As
will be clear later, this will help save time and space in the
preprocessing. We also define $V_g^1(s,G_1)$ in the same way as
before.

\begin{lemma}
For any two query points $s$ and $t$, if there does not exist a
trivial shortest \st\ path, then there is a shortest \st\ path containing a
gateway of $s$ and a gateway of $t$.
\end{lemma}
\begin{proof}
Suppose there does not exist a trivial shortest \st\ path. Then, there
is a shortest \st\ path that contains a polygon vertex. Therefore, to
prove the lemma, it is sufficient to show the following: For any polygon vertex
$p$, there exists a shortest path from $s$ (resp., $t$) to $p$ that
contains a gateway of $s$ (resp., $t$). For the case of $s$, since its
gateway set is the same as before in the graph $G$, this has been
proved in~\cite{ref:ChenSh00}. For the case of $t$, we can follow the
similar analysis as in~\cite{ref:ChenTw16} (i.e., the proof of Lemma~2)
because our way of defining $V_g^2(t,G_1)$ is similar in spirit to
theirs (the only difference is that the size of the gateway set in
\cite{ref:ChenTw16} is $O(\sqrt{\log n})$, which is due to that each
super-level of $\calT$ in \cite{ref:ChenTw16} consists of $\sqrt{\log
n}$ levels). We omit the details. \qed
\end{proof}


\begin{lemma}\label{lem:gatewaypreprocess}
With $O(n\log^3 n/\log\log n)$ time and $O(n\log^2 n/\log\log n)$
space preprocessing, we can compute $V_g(s,G_1)$ and $V_g(t,G_1)$ in
$O(\log n)$ time for any two query points $s$ and $t$.
\end{lemma}
\begin{proof}
We first discuss the case for $t$. For computing $V_g^1(t,G_1)$, as
in~\cite{ref:ChenTw16} (see the proof of Lemma~3), it is sufficient to
compute the four projection points $t^d,t^u,t^l,t^r$ on $\partial\calP$, which can be
done in $O(\log n)$ time by using the horizontal and vertical
decompositions of $\calP$. The two decompositions of $\calP$ can be
computed in $O(n\log n)$ time or $O(n+h\log^{1+\epsilon}h)$ time for
any $\epsilon>0$~\cite{ref:Bar-YehudaTr94,ref:ChazelleTr91}.

For $V_g^2(t,G_1)$, we can use the same approach as that
in~\cite{ref:ChenTw16} (see the proof of Lemma~3). Since the number of
type-3 Steiner points is $O(n\log^2n/\log\log n)$, the preprocessing takes
$O(n\log^3n/\log\log n)$ time and  $O(n\log^2n/\log\log n)$ space.

For $s$, the set $V_g^1(s,G_1)$ can be computed in the same way as
$t$. For $V_g^2(s,G_1)$, we maintain a data structure for all type-2
Steiner points as in~\cite{ref:ChenTw16} (see the proof of Lemma~3).
Since there are $O(n\log n)$ type-2 Steiner points,
with $O(n\log^2 n)$ time and $O(n\log n)$ space preprocessing,
we can compute $V_g^2(s,G_1)$ in $O(\log n)$ time.
\qed
\end{proof}

Our preprocessing is similar as before. For each vertex
$q$ of $G_1$ (which is also considered as a point in $\calP$),
we compute a shortest path tree $T(q)$ but only for the
points in $\calV_1\cup\calV_2$ using the algorithm~\cite{ref:MitchellAn89,ref:MitchellL192}.
Since $|\calV_1\cup\calV_2|=O(n\log n)$, $T(p)$ has
$O(n\log n)$ vertices and can be computed in $O(n\log^2n)$ time~\cite{ref:MitchellAn89,ref:MitchellL192}. We
also store the post-order traversal list of $T(p)$. Since $G_1$ has
$O(n\log^2n /\log\log n)$ vertices, the preprocessing takes
$O(n^2\log^4n/\log\log n)$ time and  $O(n^2\log^3n/\log\log n)$
space in total.

\paragraph{Remark.} If we define $V_g^2(s,G_1)$ in the same way as $V_g^2(t,G_1)$ (i.e.,
with respect to type-3 Steiner points), then we would need to compute $T(p)$ for
all $O(n\log^2 n/\log\log n)$ type-3 Steiner points in the
preprocessing, which would take $O(n^2\log^5n/(\log\log n)^2)$ time
and $O(n^2\log^4n/(\log\log n)^2)$ space.
\vspace{0.1in}

As a summary, we have the following result.

\begin{lemma}\label{lem:query}
With $O(n^2\log^4n/\log\log n)$ time and  $O(n^2\log^3n/\log\log n)$
space preprocessing, given any two query points $s$ and $t$, we can
compute their shortest path length in $O(\log n)$ time and an
actual shortest \st\ path can be output in additional time linear in the number of
edges of the path.
\end{lemma}
\begin{proof}
With the new gateway sets $V_g(s,G_1)$ and  $V_g(t,G_1)$,
applying Lemma~\ref{lem:query22} directly will lead to the lemma.
To guarantee correctness, since we now use new gateway sets $V_g^2(s,G_1)$
and  $V_g^2(t,G_1)$, we need to show that the geometric properties in
Section~\ref{sec:subproblem} related to these gateways still hold.
Specifically, we need to show that the properties of the gateway
region $R(s)$ for $s$, i.e., Observations~\ref{obser:gatewayregion}
and~\ref{obser:rectangle}, and the properties of the extended gateway
region $R(t)$ for $t$, i.e., Observations~\ref{obser:negative},
\ref{obser:posabove}, \ref{obser:postive}, 
\ref{obser:qhneg}, and \ref{obser:qhpos}, still hold. Indeed, for $R(s)$, its properties
obviously hold since $R(s)$ is exactly the same as before (because
$V_g^2(s,G_1)$ is exactly $V_g^2(s,G)$). For $R(t)$, its properties also hold.
An easy way to see this is that the new $R(t)$ defined based on
$V_g^2(t,G_1)$ is a subset of the original $R(t)$ defined based on
$V_g^2(t,G)$.
\qed
\end{proof}

\subsection{A Further Improvement}

Using the techniques in~\cite{ref:ChenTw16}, we can further
reduce the complexities of the preprocessing so that they are
functions of $h$, in addition to $O(n)$,
as shown in the following theorem.

\begin{theorem}
With $O(n+h^2\log^4 h/\log\log h)$ time and $O(n+h^2\log^3 h/\log\log h)$ space preprocessing, given any two query points $s$ and $t$, we can compute their shortest path length in $O(\log n)$ time and an actual shortest \st\ path can be output in additional time linear in the number of edges of the path.
\end{theorem}

The main idea is to follow the algorithmic scheme in~\cite{ref:ChenTw16} (i.e., the one in Section~4), by replacing the ``enhanced'' graph $G_E$ with our graph $G_1$ and replacing their query algorithm with our new query algorithm. A major difference is that since our query algorithm needs to determine the relative positions of two shortest paths, we will also need to compute (planar) shortest path trees using the algorithms in~\cite{ref:ChenCo19} (we cannot use the shortest path trees in $G_1$ because they may not be planar).
We only sketch the main idea below, following the notation in~\cite{ref:ChenTw16}.

The algorithm in~\cite{ref:ChenTw16} uses an extended corridor structure
to decompose $\calP$ into an {\em ocean} $\calM$, and $O(n)$ {\em bays} and {\em canals}.
While $\calM$ is multiply-connected, each bay/canal is a simple polygon. Each bay has a {\em gate} which is a common edge shared by the bay and $\calM$. Each canal has two gates.

A graph $G_E(\calM)$ is built on $\calM$ with respect to $O(h)$ special
points on the boundary of $\calM$. The graph has $O(h\sqrt{\log
h}2^{\sqrt{\log h}})$ vertices and edges. Using the graph, with
$O(n+h^2\log^2 h4^{\sqrt{\log h}})$ time and $O(n+h^2\log
h4^{\sqrt{\log h}})$ space preprocessing, if $s$
and $t$ are both in $\calM$, a shortest \st\ path can be computed in
$O(\log n)$ time.

For our purpose, we replace the graph $G_E(\calM)$ by our graph $G_1(\calM)$, which is built with respect to the above mentioned $O(h)$ special points on the boundary of $\calM$ in the same way as the graph $G_1$ with respect to the obstacle vertices of $\calP$. Thus, the graph $G_1(\calM)$ has
$O(h\log^2 h/\log\log h)$ vertices and edges. We define the sets $\calV_1(\calM)$ and $\calV_2(\calM)$ accordingly (in the same way as $\calV_1$ and $\calV_2$ defined for $G_1$), which together have $O(h\log h)$ points.
For each vertex $q$ of $G_1(\calM)$, we need to compute a planar shortest path tree $T(q)$
from $q$ to all points of $\calV_1(\calM)\cup \calV_2(\calM)$. To this end, if we applied the algorithms
in~\cite{ref:MitchellAn89,ref:MitchellL192}
as before, then the total preprocessing would take $\Omega(nh)$ time
and space. To make the preprocessing complexities linearly depend on
$n$, we use an algorithm in~\cite{ref:ChenCo19} instead.
Specifically, we can compute a shortest path tree $T'(q)$ among all {\em cores} of the obstacles of $\calP$ (see \cite{ref:ChenCo19} for the details). Since the size of $\calV_1(\calM)\cup \calV_2(\calM)$ is $O(h\log h)$, $T'(q)$ can be computed in $O(h\log^2
h)$ time and $O(h\log h)$ space~\cite{ref:ChenCo19}.
In particular, each vertex of $T'(q)$ is also a vertex of $T(q)$, and
the length of a path in $T'(q)$ from $q$ to any vertex $p$ is equal to
$d(p,q)$. Although an edge of $T'(q)$ may not be
in $\calP$, the paths in $T'(q)$ maintain the same topology as those in
$T(q)$, i.e., the relative positions of two paths from $q$ to two
vertices in $T'(q)$ are consistent with those in $T(q)$
(e.g., this can be seen from the proof of Lemma~2 in~\cite{ref:ChenCo19}). Therefore, we
can use $T'(q)$ to determine the relative positions of two shortest paths in our query algorithm. In this way, with $O(n+h^2\log^4h/\log\log h)$ time and $O(n+h^2\log^3h/\log\log h)$
space preprocessing, we can compute the shortest path length $d(s,t)$ in $O(\log n)$ time.

However, we are not able to output a shortest \st\ path in additional time linear in the
number of edges of the path since a path in $T'(q)$ may not be in
$\calP$ (although we can compute a shortest \st\ path in $O(n)$
additional time, e.g., see the proof of Lemma~2
in~\cite{ref:ChenCo19}). To resolve this issue, we use the following approach.
Note that the query algorithm will return a gateway $p$ of $s$ and a gateway $q$ of $t$ so that there is a shortest \st\ path containing both $p$ and $q$. Our goal is to find a shortest path in $\calP$ from $p$ to $q$ (and then by appending $\overline{sp}$ and $\overline{tq}$, we can obtain a shortest \st\ path). To this end, we will build another graph $G_2(\calM)$ with the following properties: (1) $G_2(\calM)$ has $O(h\log^2 h/\log\log h)$ vertices and edges, the same as in $G_1(\calM)$ asymptotically; (2) each vertex of $G_1(\calM)$ is also a vertex in $G_2(\calM)$; (3) for any two vertices $u$ and $v$ of $G_2(\calM)$ that are also vertices of $G_1(\calM)$, a shortest path from $u$ to $v$ in $G_2(\calM)$ corresponds to a shortest path in the plane with the same length. We will discuss the definition of $G_2(\calM)$ later in Section~\ref{sec:pathpre}.

With $G_2(\calM)$, to find a shortest path from $p$ to $q$, if we compute a shortest path tree $T''(p)$ in $G_2(\calM)$ from $p$ to all vertices of $G_2(\calM)$ in the preprocessing, then we can report the path in $T''(p)$ from $p$ to $q$ as a shortest path in time linear in the number of edges of the path.
Observe that $p$, which is a gateway of $s$, is a point in $\calV_1(\calM)\cup \calV_2(\calM)$.
Correspondingly, in the preprocessing, for each point $v\in \calV_1(\calM)\cup \calV_2(\calM)$, which is also a vertex of $G_2(\calM)$, we compute a shortest path tree $T''(v)$ in $G_2(\calM)$ from $v$ to all vertices of $G_2(\calM)$. Since $|\calV_1(\calM)\cup \calV_2(\calM)|=O(h\log h)$ and $G_2(\calM)$ has $O(h\log^2 h/\log\log h)$ vertices and edges, computing all these shortest path trees takes $O(h^2\log^4 h/\log\log h)$ time and $O(h^2\log^3 h/\log\log h)$ space. In addition, as in~\cite{ref:ChenTw16}, we need $O(n)$ space to store ``corridor paths'' and ``elementary curves'' (see \cite{ref:ChenTw16} for the details), and these information will also be used to output actual shortest paths.

The above discusses the case where both $s$ and $t$ are in the ocean
$\calM$. To process the queries for other cases (i.e., at least one of $s$ and $t$ is not in $\calM$), the algorithm in~\cite{ref:ChenTw16} builds an additional graph $G_E(g)$ of similar structures for each gate $g$ of a canal or a bay. Then, the graph $G_E(\calM)$ is merged with all these additional graphs $G_E(g)$ to obtain a graph $G_E(\calP)$,
which has $O(h\sqrt{\log h}2^{\sqrt{\log h}})$ vertices and edges, the same as $G_E(\calM)$
asymptotically. Using $G_E(\calP)$,  with
$O(n+h^2\log^2 h4^{\sqrt{\log h}})$ time and $O(n+h^2\log
h4^{\sqrt{\log h}})$ space preprocessing, each query can be answered in $O(\log n)$ time.

For our purpose, we replace each graph $G_E(g)$ correspondingly by our graph $G_1(g)$ and then
obtain a new merged graph $G_1(\calP)$, which has $O(h\log^2
h/\log\log h)$ vertices and edges. We define  $\calV_1(\calP)$ and $\calV_2(\calP)$ accordingly, which together have $O(h\log h)$ points. Also, for each vertex $p$ of
$G_1(\calP)$, we compute a shortest path tree $T'(p)$ for all points of $\calV_1(\calP)\cup \calV_2(\calP)$.
This can be done in the same time
and space as before in $\calM$ asymptotically. The total preprocessing
time and space are $O(n+h^2\log^4h/\log\log h)$ and
$O(n+h^2\log^3h/\log\log h)$, respectively.
For any two query points $s$ and $t$, we can apply the query algorithm
scheme in~\cite{ref:ChenTw16} along with our new query algorithm in Lemma~\ref{lem:query22} to
compute $d(s,t)$ in $O(\log n)$ time.
For reporting an actual shortest \st\ path, we use the similar approach as above for the ocean case but instead use a graph $G_2(\calP)$, by merging $G_2(\calM)$ with $G_2(g)$ for all gates $g$.

In summary, with $O(n+h^2\log^4 h/\log\log h)$ time and $O(n+h^2\log^3 h/\log\log h)$ space preprocessing, given $s$ and $t$, we can compute $d(s,t)$ in $O(\log n)$ time and an actual shortest \st\ path can be output in additional time linear in the number of edges of the path.

\subsubsection{A path-preserving graph}
\label{sec:pathpre}

It remains to define the graph $G_2(\calM)$. To do so, we define a graph $G_2$ based on $G_1$ with respect to all polygon vertices of $\calP$ (and thus $G_2(\calM)$ has a similar structure but based on $G_1(\calM)$).

As discussed before, although $G_1$ preserves shortest paths among all
polygon vertices, it may not preserve shortest paths for all its
vertices. Our goal is to modify $G_1$ to obtain another graph $G_2$,
so that each vertex of $G_1$ is also in $G_2$ and $G_2$ preserves
shortest paths for all vertices of $G_1$. A straightforward way to
do so is to build a graph $G'$ with respect to all vertices of $G_1$
in the same way as we build $G_1$ with respect to all polygon
vertices. However, since $G_1$ has $O(n\log^2 n/\log\log n)$ vertices, such a graph
$G'$ would have $O(n\log^3 n/\log\log n)$ vertices and edges. In contrast, our
graph $G_2$ only has $O(n\log^2 n/\log\log n)$ vertices and edges, the same as in $G_1$
asymptotically.

Recall that $\calV_1$ consists of all polygon vertices as well as
their projections on $\partial\calP$. We define $\calV_3$ as the set consisting of all
type-3 Steiner points of $G_1$. Hence, $\calV_2\subseteq \calV_3$ and $\calV_1\cup\calV_3$ constitutes the vertex set
of $G_1$.

Suppose we already have the graph $G_1$. We change it through the following
three steps.

First, for each point $p\in \calV_1\cup \calV_3$ and each of $p$'s
projection $q$ on $\partial\calP$, if $q$ is not in $\calV_1$, we
include $q$ as a new type-1 Steiner point and insert it to $G_1$,
i.e., make $q$ as a new vertex, add an edge connecting $q$ to $p$ and
two edges connecting
$q$ to its two adjacent Steiner points on the polygon edge containing
$q$. Since $|\calV_1\cup \calV_3|=O(n\log^2 n/\log\log n)$, the above adds $O(n\log^2
n/\log\log n)$ vertices and edges.

Second, for each point $p\in \calV_1$ that is not a polygon vertex, we define Steiner points on the cut-lines of $\calT$ following the same rule as before for type-2 (not type-3) Steiner points.
Specifically, if $p$ is on a cut-line (this happens when the cut-line is through a polygon
vertex such that $p$ is a vertical projection of the vertex on $\partial\calP$), then the cut-line is already a leaf $u$ of $\calT$;
otherwise, we add a cut-line through $p$ and insert it as a new leaf
$u$ in $\calT$ by the $x$-coordinate. In either case,
for each node $v$ of $\calT$ in the path from $u$ to the root, we let $p$ define a type-2 Steiner point $p'$ on
$l(v)$ if $p$ is horizontally visible to $l(v)$ and then add two edges
connecting $p'$ to its two adjacent visible Steiner points on $l(v)$.
Since $|\calV_1|=O(n)$, the above adds $O(n\log n)$ vertices and
edges.

Third, for each point $p\in \calV_1$, let $S(p)$ denote the set of all
Steiner points on all cut-lines defined by $p$, including $p$ itself
as well as $p^l$ and $p^r$.
Clearly, all points of $S(p)$ are on the segment $\overline{p^lp^r}$.
The current graph has an edge connecting $p$ to each point of
$S(p)\setminus\{p\}$, and we remove such edges and instead add an edge
to connect each pair of adjacent points of $S(p)$ from left to right.
This does not change the number of edges of the graph.

The resulting graph is $G_2$, which still has $O(n\log^2 n/\log\log
n)$ vertices and edges. In particular, the following observation is
guaranteed by the above first step.

\begin{observation}\label{obser:G2}
For each point $p\in \calV_1\cup \calV_3$, its four projections on $\partial\calP$ are all Steiner points (and thus vertices) of $G_2$.
\end{observation}

We can still construct $G_2$ in $O(n\log^3 n/\log\log n)$ time in a similar way as before.
The following lemma shows that $G_2$ preserves shortest paths for all points of $\calV_1\cup \calV_3$ (i.e., all vertices of $G_1$).

\begin{lemma}\label{lem:preserve}
For any two points $p$ and $q$ of $\calV_1\cup \calV_3$,
a shortest path from $p$ to $q$ in $G_2$ is also a shortest path in $\calP$.
\end{lemma}
\begin{proof}
Because every polygon vertex is in $\calV_1\cup \calV_3$, to prove the
lemma, following the proof scheme
in~\cite{ref:ClarksonRe87,ref:ClarksonRe88}, it is sufficient to show
the following: For any two points $p$ and $q$ in $\calV_1\cup \calV_3$
that are visible to each other, $G_2$ must have an $xy$-monotone path
connecting $p$ and $q$ if the connected component of $R\cap \calP$
containing $\overline{pq}$ does not contain any polygon vertex, where
$R$ is the rectangle with $\overline{pq}$ as a diagonal.
In the following, we assume that $p$ and $q$ are visible to each other
and $R'$ does not contain any polygon vertex, where $R'$ is the
connected component of $R\cap \calP$ containing $\overline{pq}$. Our goal is to show that $G_2$ has an $xy$-monotone path connecting $p$ and $q$.
Without loss of generality, we assume that $p$ is to southwest of $q$.

Note that each of $p$ and $q$ is defined by a point in $\calV_1$, and each of them is contained in a cut-line of $\calT$.
Let $v_p$ and $v_q$ be the points in $\calV_1$ defining $p$ and $q$, respectively.
Let $l_p$ and $l_q$ be the cut-lines containing $p$ and $q$, respectively. Each of $l_p$ and $l_q$ is stored in a node of the cut-line tree $\calT$, and we let $l$ be the cut-line in the lowest common ancestor of the two nodes storing $l_p$ and $l_q$. Hence, $l$ is between $l_p$ and $l_q$.
Depending on whether the rectangle $R$ is in $\calP$, there are two cases.

\begin{figure}[t]
\begin{minipage}[t]{0.49\linewidth}
\begin{center}
\includegraphics[totalheight=1.0in]{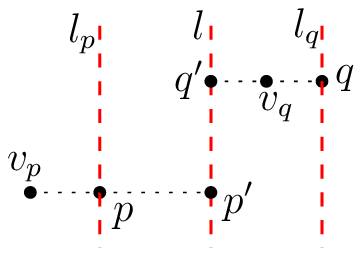}
\caption{\footnotesize
Illustrating the case $R\subseteq \calP$.}
\label{fig:monotone1}
\end{center}
\end{minipage}
\begin{minipage}[t]{0.49\linewidth}
\begin{center}
\includegraphics[totalheight=1.0in]{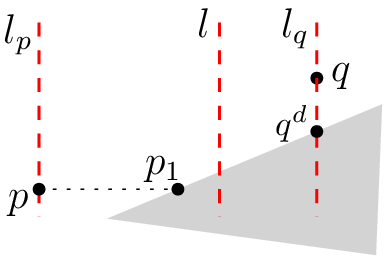}
\caption{\footnotesize
Illustrating the case where $p$ is not horizontally visible to $l$.}
\label{fig:monotone3}
\end{center}
\end{minipage}
\hspace{0.05in}
\vspace*{-0.15in}
\end{figure}

If $R\subseteq \calP$, then since $v_p$ is horizontally visible to
$l_p$, $v_p$ is also horizontally visible to $l$. Since $l$ is an
ancestor of $l_p$ in $\calT$, $v_p$ defines a Steiner point $p'$ on
$l$. For the same reason, $v_q$ also defines a Steiner point $q'$ on
$l$ (e.g., see Fig.~\ref{fig:monotone1}). Due to the third step for
changing $G_1$ to obtain $G_2$, the path
$\overline{pp'}\cup\overline{p'q'}\cup\overline{q'q}$ is in $G_2$,
which is $xy$-monotone.

If $R\not\subseteq \calP$, then if both $v_p$ and $v_q$ are still
horizontally visible to $l$, we can use the same analysis as above.
Otherwise, without loss of generality, we assume that $p$ is not
horizontally visible to $l$. This implies that if we move from $p$
rightwards following the lower edge of $R$, we will encounter a point
$p_1$ at a polygon edge $e$ before we arrive at $l$ (e.g.,
see Fig.~\ref{fig:monotone3}). Note that $p_1$
is actually the right projection of $v_p$ on $\partial\calP$ and thus
is a type-1 Steiner point by Observation~\ref{obser:G2}.
Since $R'$ does not contain any polygon vertex, the downward
projection $q^d$ of $q$ on $\partial\calP$ is on $e$ as well. By Observation~\ref{obser:G2}, $q^d$ is
a type-1 Steiner point. Note that the slope of $e$ must be positive.
Hence, the path $\overline{pp_1}\cup \overline{p_1q^d}\cup
\overline{q^dq}$, which is in $G_2$, is $xy$-monotone.

The lemma is thus proved.
\qed
\end{proof}

\section{Concluding Remarks}

In this paper, we present a data structure that can answer two-point $L_1$ shortest path queries in a polygonal domain $\calP$ in $O(\log n)$ time, and our preprocessing takes nearly quadratic time and space in the number of holes of $\calP$ plus linear time and space in the total number of vertices of $\calP$. More importantly and interestingly, we propose a divide-and-conquer algorithm that can compute a shortest path in nearly linear time in the number of gateways of $s$ and $t$, improving the previously best and straightforward quadratic time algorithm.

To further improve our result, it might be tempting to see whether the Monge matrix searching techniques~\cite{ref:AggarwalGe87,ref:KlaweAn90} can be applied so that the query time becomes linear in the number of gateways of the query points. However, due to those non-ideal situations such as those illustrated in Fig.~\ref{fig:dif1} and Fig.~\ref{fig:dif2}, it is not clear to us whether it is possible to do so.

One may wonder whether our divide-and-conquer technique can be applied to the Euclidean case. Indeed, the algorithms in both \cite{ref:ChenOn01} and \cite{ref:GuoSh08} for Euclidean two-point shortest path queries are based on the gateway approach (the gateways are called ``critical cites'' in \cite{ref:GuoSh08}). More specifically, the method of Chen et al.~\cite{ref:ChenOn01} uses the set $Q_s$ of vertices of $\calP$ visible to $s$ as the gateway set of $s$, and similarly, the set $Q_t$ of vertices of $\calP$ visible to $t$ are used as the gateway set of $t$. Without loss of generality, assume $|Q_s|\leq |Q_t|$. After $Q_s$ is computed, for each vertex $v\in Q_s$, a shortest path from $s$ to $t$ through $v$ is found by using the shortest path map of $v$ (the map is computed in the preprocessing). In this way, the query time is bounded by $O(\min\{|Q_s|,|Q_t|\}\cdot \log n)$. By using a tessellation of $\calP$, Guo et al.~\cite{ref:GuoSh08} showed that a subset of $Q_s$ of size $O(h)$ is sufficient to serve as the gateway set of $s$, and the same holds for $t$. Consequently, the query time can be bounded by $O(h\log n)$. Clearly, the bottleneck of the query time is actually on the number of gateways. To have any hope of achieving a polylogarithmic time query algorithm using gateways, one has to make sure that the number of gateways is polylogarithmic. We are able to achieve this by using path preserving graphs in the $L_1$ metric. Such graphs, however, are not applicable to the Euclidean metric. Note that the polylogarithmic time query algorithms by Chiang and Mitchell~\cite{ref:ChiangTw99} are based on different techniques (e.g., shortest path map equivalence decompositions) than using gateways. Therefore, for solving the Euclidean two-point shortest path queries, one direction is to see whether it is possible to use only a polylogarithmic number of gateways.




\bibliographystyle{plain}
\bibliography{reference}

%


\end{document}